\documentclass[12pt]{article}

\usepackage{amsmath,amssymb,amsthm,amscd,a4wide, amscd}
\usepackage{fancybox, ascmac}
\usepackage[hidelinks]{hyperref}
\usepackage{amsfonts,latexsym,makeidx}
\usepackage{amsxtra}
\usepackage{graphicx}
\usepackage{graphicx,color}
\usepackage{ulem} 
\usepackage{bm}
\usepackage{subcaption}

\usepackage{eso-pic}

\makeatletter
\DeclareFontFamily{U}{tipa}{}
\DeclareFontShape{U}{tipa}{m}{n}{<->tipa10}{}
\newcommand{\arc@char}{{\usefont{U}{tipa}{m}{n}\symbol{62}}}%

\newcommand{\arc}[1]{\mathpalette\arc@arc{#1}}

\newcommand{\arc@arc}[2]{%
  \sbox0{$\m@th#1#2$}%
  \vbox{
    \hbox{\resizebox{\wd0}{\height}{\arc@char}}
    \nointerlineskip
    \box0
  }%
}
\makeatother

\newcommand{\De}{\Delta}

\newcommand{\vp}{\varphi}

\newcommand{\om}{\omega}

\newcommand{\tss}{\hspace{1pt}}

\newcommand{\CC}{\mathbb{C}\tss}

\newcommand{\ZZ}{\mathbb{Z}\tss}

\renewcommand{\theequation}{\arabic{section}.\arabic{equation}}

\newtheorem{thm}{Theorem}[section]
\newtheorem{lemma}[thm]{Lemma}
\newtheorem{prop}[thm]{Proposition}
\newtheorem{cor}[thm]{Corollary}
\newtheorem{conj}[thm]{Conjecture}

\newtheorem{defprop}[thm]{Definition-Proposition}
\newtheorem{defthm}[thm]{Definition-Theorem}

\theoremstyle{definition}
\newtheorem{definition}[thm]{Definition}

\newtheorem{remark}[thm]{Remark}
\newtheorem{example}[thm]{Example}

\newcommand{\bth}{\begin{thm}}
\renewcommand{\eth}{\end{thm}}
\newcommand{\bpr}{\begin{prop}}
\newcommand{\epr}{\end{prop}}
\newcommand{\ble}{\begin{lem}}
\newcommand{\ele}{\end{lem}}
\newcommand{\bco}{\begin{cor}}
\newcommand{\eco}{\end{cor}}
\newcommand{\bde}{\begin{defin}}
\newcommand{\ede}{\end{defin}}
\newcommand{\bex}{\begin{example}}
\newcommand{\eex}{\end{example}}
\newcommand{\bre}{\begin{remark}}
\newcommand{\ere}{\end{remark}}
\newcommand{\bcj}{\begin{conj}}
\newcommand{\ecj}{\end{conj}}

\newcommand{\bal}{\begin{aligned}}
\newcommand{\eal}{\end{aligned}}
\newcommand{\beq}{\begin{equation}}
\newcommand{\eeq}{\end{equation}}
\newcommand{\ben}{\begin{equation*}}
\newcommand{\een}{\end{equation*}}

\newcommand{\bpf}{\begin{proof}}
\newcommand{\epf}{\end{proof}}

\def\beql#1{\begin{equation}\label{#1}}

\newcommand{\RR}{\mathbb{R}}

\newcommand{\nln}{\nonumber\\}

\newcommand{\ms}{\medskip}

\newcommand{\al}{\alpha}
\newcommand{\be}{\beta}
\newcommand{\ga}{\gamma}
\newcommand{\ep}{\epsilon}
\newcommand{\de}{\delta}
\newcommand{\ka}{\kappa}
\newcommand{\la}{\lambda}

\newcommand{\dis}[1]{$\displaystyle{#1}$}

\newcommand{\ket}[1]{ | {#1} \rangle}
\newcommand{\vev}[1]{ \langle {#1} \rangle}

\newcommand{\bg}{\bm{g}}
\newcommand{\ee}{\bm{e}}
\newcommand{\bnu}{\bm{\nu}}

\newcommand{\bA}{\bm{A}}

\newcommand{\secref}[1]{\S\,\ref{#1}}
\newcommand{\appref}[1]{App.\,\ref{#1}}
\newcommand{\ssecref}[1]{\S\,\ref{#1}}
\newcommand{\figref}[1]{Fig.\,\ref{#1}}
\newcommand{\tabref}[1]{Tab.\,\ref{#1}}
\newcommand{\thmref}[1]{Thm.\,\ref{#1}}
\newcommand{\propref}[1]{Prop.\,\ref{#1}}
\newcommand{\corref}[1]{Cor.\,\ref{#1}}
\newcommand{\defref}[1]{Def.\,\ref{#1}}
\newcommand{\remref}[1]{Rem.\,\ref{#1}}
\newcommand{\lemref}[1]{Lem.\,\ref{#1}}
\newcommand{\dpropref}[1]{Def.-Prop.\,\ref{#1}}
\newcommand{\dthmref}[1]{Def.-Thm.\,\ref{#1}}

\title{\Large\bf 
Rotation angles of a rotating disc  \\ 
-- A toy model exhibiting the geometric phase --
} 

\author{Takuya Matsumoto$^\ast$\,,\quad 
Hiroki Takada$^\dagger$\,, \quad 
Osami Yasukura$^\ddagger$}

\date{} 

\begin{document}

\AddToShipoutPictureFG*{%
\put(360,750){%
\parbox[t]{7cm}{%
\raggedleft
\small
Published version:\\
\textit{FORMA} \textbf{41} (2026), 13--30\\
DOI: 10.55653/forma.2026.001.002
}}
}

\maketitle

\noindent
{\it
${}^{\ast, \ddagger}$
Department of Applied Physics, 
Faculty of Engineering, 
University of Fukui, 
$3$-$9$-$1$ Bunkyo, Fukui-shi, Fukui $910$-$8507$, Japan
}
\\
{\it
${}^{\ast}$
Department of 
Fundamental Engineering for Knowledge-Based Society, 
Graduate School of Engineering, University of Fukui, 
$3$-$9$-$1$ Bunkyo, Fukui-shi, Fukui $910$-$8507$, Japan
}
\\
{\it
${}^{\dagger}$
Department of Human and Artificial Intelligent Systems (HART), 
Graduate School of Engineering, University of Fukui, 
$3$-$9$-$1$ Bunkyo, Fukui-shi, Fukui $910$-$8507$, Japan
}
\\
~\\
E-mail: {\tt \{$^\ast$takuyama, $^\dagger$takada, $^\ddagger$yasukura\}@u-fukui.ac.jp} 

\vspace{20mm}

\begin{abstract}
In this paper, we consider a simple kinematic model, which is a rotating disc on the edge of another fixed disc without slipping, and study the rotation angle of the rotating disc. The rotation angle consists of two parts, the dynamical phase $\Delta_d$ and the geometric phase $\Delta_g$. The former is a dynamical rotation of the disc itself, and the geometric motion of the disc characterizes the latter. In fact, $\Delta_g$ is regarded as the geometric phase appearing in several important contexts in physics. The clue to finding the explicit form of $\Delta_g$ is the Baumkuchen lemma, which we called. Due to the Gauss-Bonnet theorem,  in the case that the rotating disc comes back to the initial position, $\Delta_g$ is interpreted as the signed area of a two-sphere enclosed by the trajectory of the Gauss vector, which is a unit normal vector on the moving disc. We also comment on typical models sharing the common underlying structure, which include Foucault's pendulum, Dirac's monopole potentials, and Berry phase. Hence, our model is a very simple but distinguished one in the sense that it embodies the essential concepts in differential geometry and theoretical physics such as the Gauss-Bonnet theorem, the geometric phase, and the fiber bundles. 
\end{abstract}



\setcounter{footnote}{0}
\setcounter{page}{0}
\thispagestyle{empty}

\newpage

\tableofcontents

\newpage

\section{Introduction}
\label{sec:int}
\setcounter{equation}{0}

It is not easy to be aware of something that exists 
around us like the air.  
The {\it geometric phase} is a phase shift 
which an object acquires when it moves in spacetime, 
and it only depends on the geometric property of the orbit.  
Though the existence is ubiquitous in physics, 
it seemed to take a long time for us to recognize its significance. 
Interestingly, the notion has been pointed out 
in pretty different contexts in physics, 
such as classical optics \cite{Pan}, 
molecular physics \cite {Lon}, classical physics \cite{han},  
and quantum mechanics \cite{Berry1,Berry2}. 
For a comprehensive review, see \cite{SW}.  
This reflects that the geometric phase is a universal 
phenomena in physics. 

\ms 

From a mathematical point of view, the geometric phase is 
regarded as the {\it holonomy} of the fiber bundles. 
Concerning the mathematical background of differential geometry 
and topology, see \cite{KN,S, Nak, EGH}. 
Hence, the difficulty in noticing the geometric phase is that 
of {\it seeing} the fiber in the real world. 
The simplest but nontrivial example of the principal fiber 
bundle is the {\it Hopf fibration} \cite{hopf}\,, which is 
$S^1$ fibration over $S^2$ and whose total space is $S^3$\,. 

\ms 

In 1931, the same year when the Hopf fibration was proposed, 
Dirac has considered the {\it magnetic monopole potential,}  
which is the vector potential describing the magnetic nomopole, 
and shown that all electric charges should be quantized 
if the magnetic monopole exists \cite{D}. 
Dirac's quantization condition is explained as the 
topological consistency condition of the {\it gauge transformation} 
of the $U(1)$ principal fiber bundle \cite{WY1, WY2}. 
For a review of Dirac's monopole, see, for instance, \cite{Cole}. 
Furthermore, the monopole potential is entirely interpreted as 
the connection one-form of the Hopf fibration \cite{Min, Ryder}. 

\ms 

In this paper, we consider a simple kinematic model, 
which is a rotating disc on the edge of  another fixed disc 
without slipping, 
and study how the moving disc has rotated during its motion
(see \figref{fig:prob}). 
We shall decompose the total rotation angle into 
two parts, the {\it dynamical phase} $\De_d$ and 
the {\it geometric phase} $\De_g$\,. 
The former is a dynamical rotation of the disc itself, 
and the geometric motion 
of the disc characterizes the latter. 
This factor $\De_g$ is indeed regarded as the geometric phase 
that appears in several essential contexts in physics, 
as mentioned above. 
The clue to finding the explicit form of $\De_g$
(\dthmref{thm:line-int}) is 
\lemref{lem:bk}\,,  
which we call {\it Baumkuchen lemma}.


\ms 

To reveal the geometric meaning of $\De_g$\,, 
we shall introduce the Gauss map (\defref{def:gmap}), 
which defines a unit normal vector on the moving disc.  
Consequently, it is given by the line integral 
of the geodesic curvature
along the curve drawn by the Gauss vector on a two-sphere 
(\thmref{thm:De-twoint}). 
Furthermore, due to the {\it Gauss-Bonnet theorem},  
in the case that the rotating disc comes back to the original position, 
$\De_g$ turns out to be 
the area of a two-sphere enclosed by the trajectory 
of the Gauss vector.  
That is the claim of our main theorem \thmref{thm:main}
in \secref{sec:model}. 

\ms 

We will also comment on typical models sharing 
the common underlying
mathematical structure, which include Foucault's pendulum, 
Dirac's monopole potentials, and Berry phase in \secref{sec:ex}. 
In particular, our approach naturally explains the 
{\it Foucault's sine law} and leads us to propose
its generalization. 
In this sense, our model is a straightforward but outstanding one 
because it embodies the essential concepts in 
differential geometry and theoretical physics
such as the Gauss-Bonnet theorem,  
the geometric phase, and the fiber bundles. 

\ms 

This article is organized as follows.
In \secref{sec:setup}, we set up the model, 
which is a rotating disc around a fixed disc, 
and explain the question that we would like to answer. 
Then we propose our fundamental problem.  
We solve the model in \secref{sec:model}, and give 
the complete answer to the fundamental 
question, which is \thmref{thm:main} in \ssecref{subsec:thm}.   
We examine how our solution works for some concrete motions
of the rotating disc in \secref{sec:ex}. 
In \secref{sec:com}, we explain that 
some essential physical models share the common 
mathematical structure, 
which contain Foucault's pendulum, 
Dirac's monopole potentials, and Berry phase. 
In this sense, our model is {\it not} an isolated one. 
The relation between the Hopf fibration and our model
is mentioned in \secref{sec:concl}.  
\appref{app:alt-proof} is devoted to an alternative proof
of \thmref{thm:main}. 
In \appref{app:baum}, we elucidate how our Baumkuchen lemma
is related to the geodesic curvature.


\section{Set up the model}
\label{sec:setup}
\setcounter{equation}{0}

Let us start to formulate our kinematical model. 
We define the model in \ssecref{subsec:def}. 
In \ssecref{subsec:dyngeo}, we introduce 
the notion of the dynamical phase and 
the geometric phase after some observations. 

\subsection{Definition of the model}
\label{subsec:def}

We consider the following model,\footnote{
This model was initially proposed as one of the
common problems 
for the 32nd Japan Mathematics Contest and 
the 25th Japan Junior Mathematics Contest \cite{JMC}.}
which we refer to as 
{\it ``Rotation angles of a rotating disc.''}

\newpage

\begin{itembox}[l]{\bf \large 
Rotation angles of a rotating disc}
Consider a disc A with the radius $a>0$ on $z=0$ plane 
({\it i.e.} $xy$-plane) in three-dimensional space $\RR^3$
and the center is located at the origin.  
Set another disc B in $\RR^3$ with the radius $b>0$
so that it contacts disc A at 
$(a\cos\theta, a\sin \theta,0)\in \RR^3$ with 
$\theta \in \RR$\,. 
Suppose that the perpendicular line through the center of disc B must be parallel to 
or cross at some point with the $z$-axis. 
Denote the angle between $z=0$ plane and the disc B by $\beta$\,, 
and suppose that $0\leq \beta \leq\pi$\,, 
see Fig.\ref{fig:prob}\,. 
The position of disc B is entirely determined by two parameters $\theta$ and $\be$\,.  

\ms 

Let disc B rotate on the edge of disc A without slipping. 

\ms

Introduce the time parameter $t\in [0,1]\subset \RR$
to describe this motion,
and regard the angle variables 
$\theta\,, \be$ as the continuous functions of $t$\,.  
Then, the position of the rigid body, disc B, 
is determined by the map 
\begin{align}
(\theta, \beta): [0,1]\to \RR \times [0, \pi]\,, 
\quad 
t\mapsto (\theta(t), \beta(t))\,. 
\label{eq:motion}
\tag{M}
\end{align}
Set $\theta(0)=0$\,. 
Starting at $t=0$, disc B returns to the initial position at $t=1$\,. 
That is 
\begin{align}
\theta(1)=2\pi n
\quad \text{with} \quad n\in \ZZ\,, 
\quad \be(1)=\be(0)\,. 
\tag{T}
\label{eq:top}
\end{align}
We shall call the integer $\theta(1)/2\pi=n$ the 
{\it topological number}. 


\end{itembox}

\begin{figure}[htbp]
\centering
\includegraphics[keepaspectratio, scale=0.4]{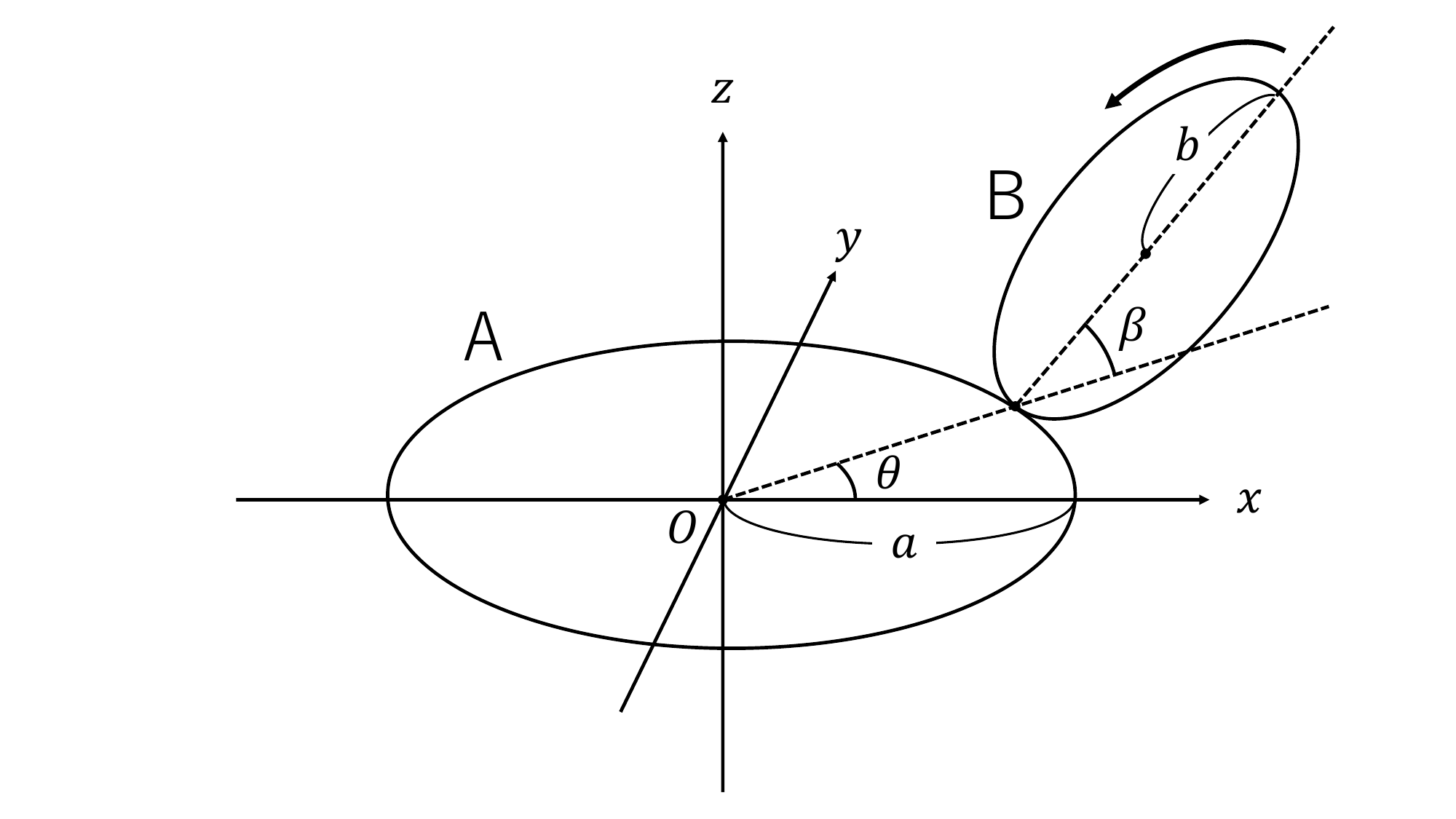}
\caption{Disc A is fixed on $xy$-plane and 
disc B 
rolls on the edge of the disc A 
without slipping. }
\label{fig:prob}
\end{figure}


\begin{remark}
\label{rem:lip}
For the functions $\theta(t)\,, \beta(t)$\,,
it is supposed to meet the following conditions. 
\begin{enumerate}
\item The continuous functions $\theta(t)\,, \beta(t)$  
satisfy the {\it Lipschitz condition}. 
That is, for any $t_1\,, t_2\in [0,1]$\,, there exists a positive 
constant $C_f$ such that 
\begin{align}
\left|f(t_1)-f(t_2)\right|\leq C_f |t_1-t_2| \,,
\end{align}
where $f$ is either $\theta$ or $\be$\,. 

\item Both $\theta(t)$ and $\beta(t)$ are differentiable except 
finite number of points on $(0,1)$.  
The derivatives $\theta'(t)\,, \be'(t)$ are piecewise continuous
on $[0,1]$\,.  


\end{enumerate}  
These two conditions are physically natural. 
Since $t\in [0,1]$ is an auxiliary parameter to describe the motion 
\eqref{eq:motion}\,, 
it is always possible to assume that  $\theta(t)\,, \be(t)$ 
satisfy the above conditions. 
\end{remark}

\ms 

Our interest is the next question. 

\ms 

\begin{itembox}[l]{\bf Question}
How much does the disc B rotate in this motion from 
$t=0$ to $t=1$? 
Define and determine the rotation angle $\De$\,. 
\end{itembox}

\medskip 

\subsection{Dynamical phase and geometric phase}
\label{subsec:dyngeo}

Before considering the generic cases, 
let us begin to examine some elementary cases. 
We will observe that the rotation angle $\De$ 
consists of two parts, {\it i.e.}, 
the {\it dynamical} phases and the {\it geometric} phases.  

\subsubsection{Observations}
\label{subsubsec:obs}

The typical cases are that the angle $\beta$ is constant
in the motion \eqref{eq:motion}\,. 
Here, we shall have a look at three cases,
(i) $\beta=0$\,, (ii) $\be=\pi/2$\,, and (iii) $\be=\pi$\,. 
\begin{itemize}
\item[(i)] In the first case $(\theta, \be)=(2\pi t, 0)$\,, 
two discs A and B are on the same plane, $z=0$\,. 
The contacting point of two discs moves counterclockwise,
starting from $(a,0,0)$ and turning back to the same point. 
The length of the contacting point traveling is $2\pi a$\,. 
Hence, dividing it by the radius of the disc B, 
we get the rotation angle $2\pi b/a$\,. 
However, in addition to this angle, 
disc B turns around disc A by $2\pi$\,,  
which equals the rotation angle of disk B 
if discs A and B are glued at the contacting point and 
they rotate together. 
Thus, the total rotation angle is given by 
\begin{align}
\De=\frac{2\pi a}{b} +2\pi \,.  
\end{align}
In particular, if two radii of discs are equal $a=b$\,, then 
it reduces to 
\begin{align}
\De=2\pi+2\pi=4\pi \,,   
\end{align}
which means that disc B is rotating {\it twice}! 
The readers can confirm this fact experimentally 
by using two coins on a desk. 
Put one coin on a desk and fix it with one hand. 
Then, turn another around the fixed one so that it does not slip on the edge.
You will see that the moving coin rotates twice when it comes back to the original position.

\item[(ii)] The second case $(\theta, \be)=(2\pi t, \pi/2)$ corresponds that disc B is vertical to disc A. 
In this case, the rotation angle of disk B is $0$  
if it is slipping at the contacting point.
Hence, since disk B is not slipping at the contacting point, it totally rotates by
\begin{align}
\De=\frac{2\pi a}{b}\,.  
\end{align}
If $a=b$\,, it simply becomes 
\begin{align}
\De=2\pi\,.  
\end{align}
This means that disc B only rotates just {\it once}. 
It is also possible to see this fact with two coins on a table, 
but it would require more concentration and careful treatment. 

\item[(iii)]
For the third case $(\theta, \be)=(2\pi t, \pi)$\,, 
though the contacting point walks the same length, 
$2\pi a$\,,  as the case (i), disc B rotates inside of disc A. 
This direction of the rotation is opposite to the first case. 
Thus, the total rotation angle is given by  
\begin{align}
\De=\frac{2\pi a}{b} -2\pi \,.  
\end{align}
When $a=b$\,, the rotation angle vanishes 
\begin{align}
\De=2\pi -2\pi=0\,. 
\end{align}
This can be explained by the two-coin experiment as follows. 
In this case, one coin is just upon another. 
Since they are the same size, 
we cannot move the coin above without it slipping. 
\end{itemize}

\subsubsection{Dynamical  phase and geometric phase}

The previous three examples imply that the total rotation angle 
$\De$ of disc B consists of two components.
One component is the factor $2\pi a/b$\,. 
This is the ratio of the length of the curve 
which the contacting point of two discs moves 
to the radius of disc B.  
We call this {\it the dynamical phase} and denote
it by $\De_d$\,. 
The dynamical phase commonly appears in any case. 
In general, it is given by 
\begin{defprop}
\label{prop:dyn}
The dynamical phase for the motion \eqref{eq:motion}
with the topological condition \eqref{eq:top} 
is given by 
\begin{align}
\De_d
=\frac{2\pi n a}{b}\,. 
\end{align}
\end{defprop}

\begin{proof}
For the motion \eqref{eq:motion} satisfying the topological 
condition \eqref{eq:top}\,, 
the signed length of the curve 
which the contacting point of two discs moves is 
\begin{align}
a \int_0^1 \frac{d\theta(t)}{dt}dt
=a (\theta(1)-\theta(0))=2\pi na\,. 
\end{align}
Dividing this by the radius $b$\,, we obtain $\De_d$\,.   
\end{proof}
Since $\theta(t)=2\pi t$ and the topological number is $n=1$
for the above examples (i), (ii), and (iii)\,, the dynamical phase reads 
$\De_d=2\pi a/b$\,. 

\ms 
While, another component is independent of the ratio $a/b$\,, 
and it only depends on how disc B has moved in $\RR^3$\,. 
In other words, this component reflects the geometric aspect of 
the motion \eqref{eq:motion}\,. 
Hence, we refer to this as {\it the geometric phase} and 
express it by $\De_g$\,. 
In summary, the total rotation angle $\De$ of disc B is 
decomposed as  
\begin{align}
\De&=\De_d+\De_g \,, \quad \text{with}\quad 
\nln
\De_d&=
\frac{2\pi n a}{b} 
~\text{\it (dynamical phase)}\,,\qquad 
\De_g=\text{\it (geometric phase)}\,. 
\label{eq:dandg}
\end{align}

\ms 

Therefore, our question essentially reduces to 
the following problem. 

\ms 

\begin{itembox}[l]{\bf Fundamental Problem}
Determine the geometric phase $\De_g$ for arbitrary motion
\eqref{eq:motion} with the topological condition 
\eqref{eq:top}\,. 
\end{itembox}

\section{Solving the model}
\label{sec:model}
\setcounter{equation}{0}

Since we have defined the model in the previous section, 
we shall solve the model. 
In \ssecref{subsec:baum}, we propose our key lemma 
\ref{lem:bk}, {\it Baumkuchen lamma}
and give the explicit formula for the geometric phase 
in \dthmref{thm:line-int}.  
Next, to elucidate the geometric meaning, we introduce 
the notion of the Gauss map (\defref{def:gmap}) and 
the regularized curve  (\defref{def:ga-reg}) in \ssecref{subsec:gmap}. 
In \ssecref{subsec:w12}, we prepare some crucial tools, 
the connections matrix and 
the geodesic curvature.  
With these preparations, in \ssecref{subsec:thm}, 
the main theorem \ref{thm:main} will be proved.

\subsection{Baumkuchen lemma}
\label{subsec:baum}

We shall adopt the following trivial fact as 
a cornerstone of our subsequent arguments, what we call 
{\it Baumkuchen lemma}\,. 

\begin{lemma}[Baumkuchen lemma]
\label{lem:bk}
Consider two fans $OAA'$ and $OBB'$ with the same center $O$\,,
and suppose that $A, B, O$ and $A', B', O$ are on the same lines. 
(See Fig {\rm\ref{fig:bk}}.)
Let $q$ be the length $AB=A'B'$\,, which is the width of
a piece of {\rm Baumkuchen}, and 
denote the lengths of 
the arcs $\arc{AA}$ and $\arc{BB'}$ by 
$L$ and $L_q$\,, respectively. 
Then, the center angle $\theta_B=\angle AOA'=\angle BOB'$
is given by 
\begin{align}
\frac{L-L_q}{q}=\theta_B\,. 
\end{align}
In particular, the {\rm Baumkuchen angle} $\theta_B$ is 
independent of the width $q$\,. 
\begin{figure}[htbp]
\begin{center}
\vspace{-4mm}
\includegraphics[width=100mm]{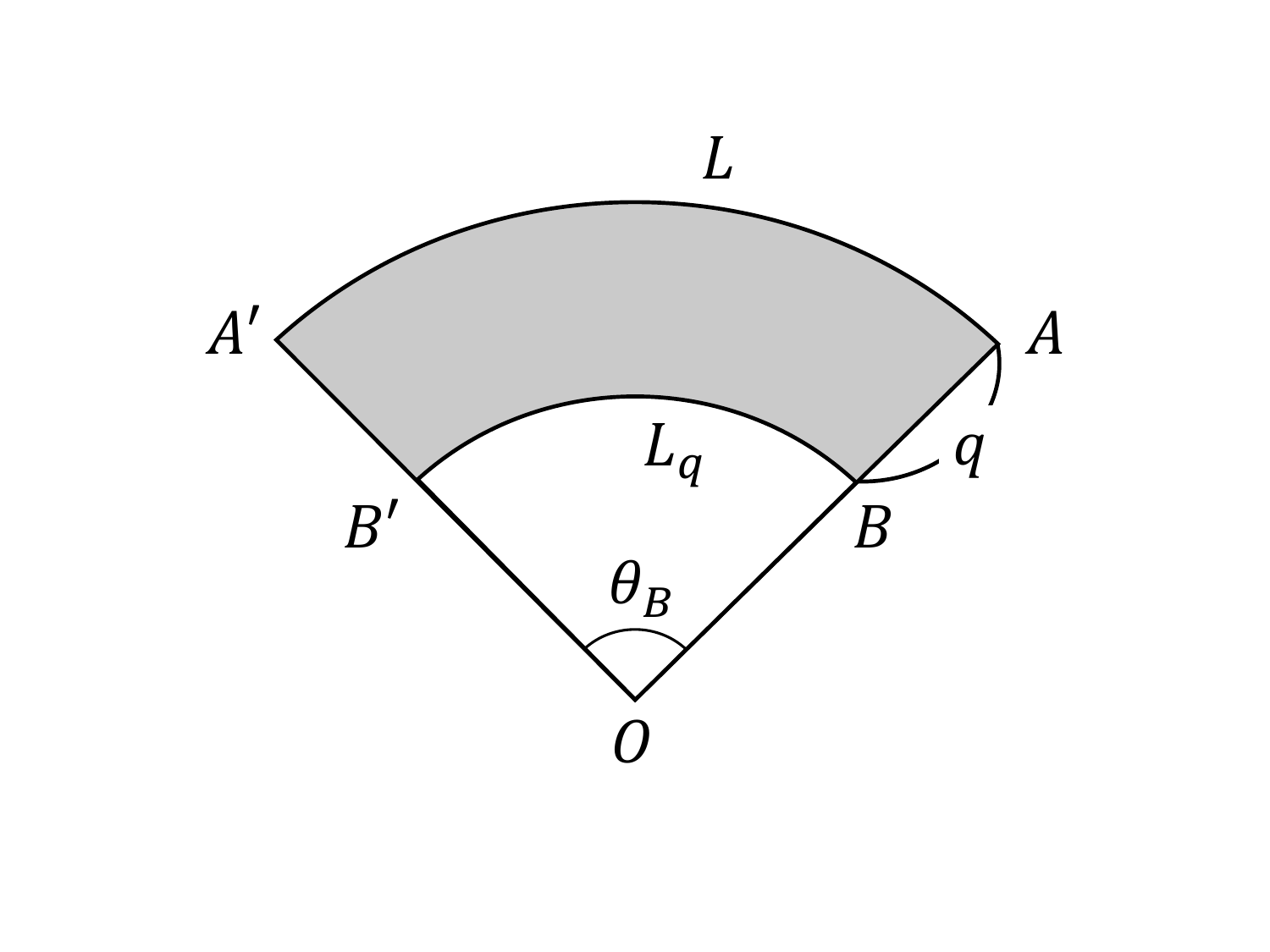}
\vspace{-13mm}
\caption{The angle $\theta_B$ of two edges $AB$ and $A'B'$
of a piece of {\it Baumkuchen} is obtained by 
dividing the difference of the arc lengths $L-L_q$ by the 
width $q$\,. }
\label{fig:bk}
\end{center}
\end{figure}
\end{lemma}

\begin{proof}
Denote the radius of the fan by $OA=OA'=r$\,. 
Since $L=r\theta_B$ and $L_q=(r-q)\theta_B$\,, 
we get
\begin{align}
\frac{L-L_q}{q}=\frac{r\theta_B-(r-q)\theta_B}{q}=\theta_B\,. 
\end{align}
This is, of course, independent of $q$\,.  
\end{proof}

\ms 

Answer for the Fundamental Problem
is obtained by using \lemref{lem:bk}. 
The purpose here is to prove the following theorem. 

\ms

\begin{defthm}
\label{thm:line-int}
The geometric phase of the motion \eqref{eq:motion} is given by 
\begin{align}
\label{eq:thm-line-int}
\De_g=\int_{0}^{1} \cos\be(t) \frac{d\theta(t)}{dt} dt \,. 
\end{align}
\end{defthm}

\ms 

\begin{proof}
First, we decompose the time interval $[0,1]$ into 
a mesh consisting of $N$ segments such that 
\begin{align}
0=t_0<t_1< \cdots < t_k <\cdots < t_{N-1} < t_N =1  \,. 
\end{align}
The typical choice of a mesh would be given by 
\begin{align}
t_k=\frac{k}{N} \qquad \text{with} \qquad 
k=0\,,1\,, \cdots, N\,. 
\end{align}

\ms 

Second, we apply Baumkuchen lemma \ref{lem:bk}
for each segment $[t_k, t_{k+1}]$ with $0\leq k\leq N-1$\,. 
On the segment $[t_k, t_{k+1}]$\,, we may assume that 
$\be(t)$ is constant for a sufficiently large $N$\,.
Let $B$ be the center of disc B, $P$ the contacting point 
of discs A and B at $t=t_k$\,, and they respectively move to 
$B'$ and $P'$ at $t=t_{k+1}$\,.   
Now, we must notice that there is a Baumkuchen with two arcs
$\arc{BB'}$ and $\arc{PP'}$\,. 
The infinitesimal motion of dics B is illustrated in 
\figref{fig:inf_bk_cf}.  
\begin{figure}
\centering
\begin{subfigure}{0.45\columnwidth}
\centering
\includegraphics[width=1.1\columnwidth]{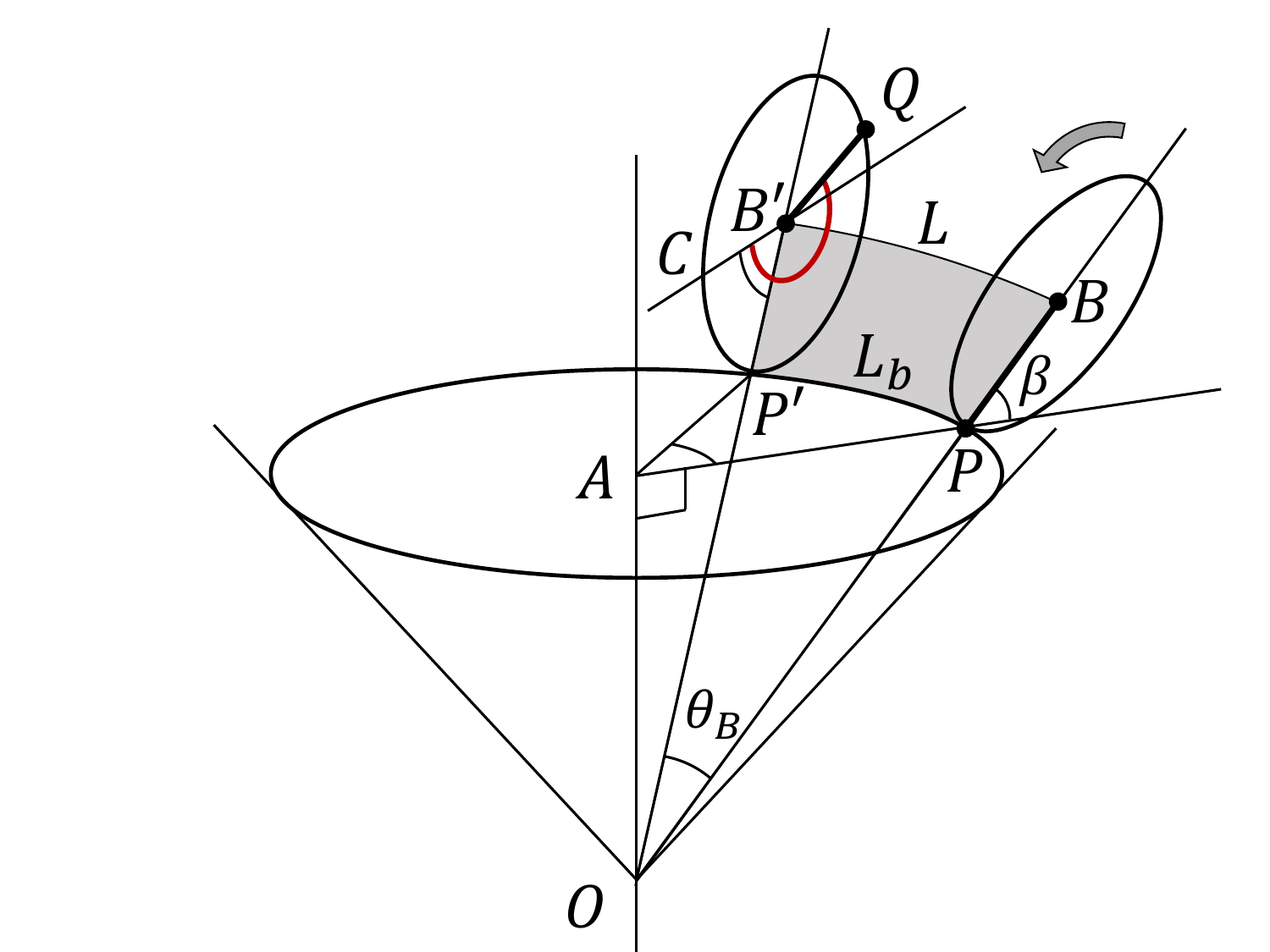}
\caption{Baumkuchen appears as a part of the cone. 
The angle $\angle PAP'$ is $\theta(t_{k+1})-\theta(t_k)$\,. 
\phantom{~~~~~~~~~~~}
}
\label{fig:inf_bk_cone}
\end{subfigure}
\hspace{3mm}
\begin{subfigure}{0.45\columnwidth}
\centering
\includegraphics[width=1.1\columnwidth]{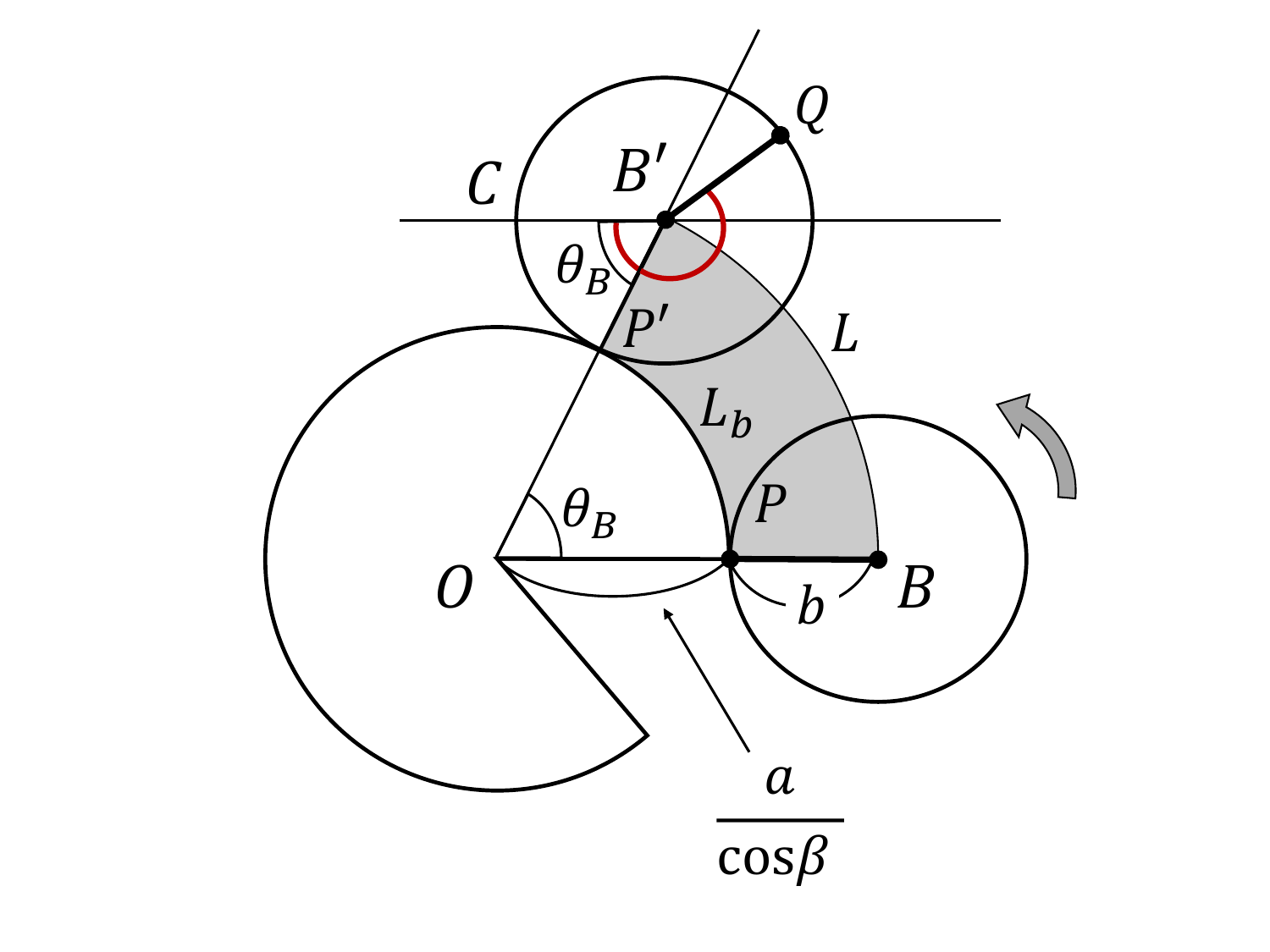}
\caption{
Baumkuchen is in the fan, an open cone in the flat space. 
Here, $CB'$ is parallel to $OB$ and hence 
$\theta_B=\angle POP'=\angle CB'P'$\,. 
}
\label{fig:inf_bk_flat}
\end{subfigure}
\caption{The infinitesimal motion of disc B in $[t_k, t_{k+1}]$\,. 
Cutting the cone in \figref{fig:inf_bk_cone} 
along the line
$OP$\,, we obtain \figref{fig:inf_bk_flat}.  
The Baumkuchen is indicated by grey. 
The total rotation angle $\angle CB'Q$ in red is given by the sum of the Baumkuchen angle 
$\theta_B$ and the dynamical angle $\angle P'B'Q$\,.
The point $P$ at $t_k$ on disc B travels to the point $Q$ at $t_{k+1}$\,.
}
\label{fig:inf_bk_cf}
\end{figure}
The infinitesimal geometric phase of disc B is identified with 
the Baumkuchen angle $\theta_B$\,. 
Since $AP=a\,, PB=b$ in \figref{fig:inf_bk_cone},  
the arc lengths $\arc{BB'}\,, \arc{PP'}$  read respectively as 
\begin{align}
L&=\arc{BB'}=(a+b\cos\be(t_k))(\theta(t_{k+1})-\theta(t_k))\,, 
\nln
L_b&=\arc{PP'}=a(\theta(t_{k+1})-\theta(t_k))\,. 
\end{align}
Here, we have regarded the radius of disc B
as the width of Baumkuchen, namely $q=b$\,.  
Thus, by the Baumkuchen lemma \ref{lem:bk}\,, 
the geometric phase can be calculated as 
\begin{align}
\theta_B=\frac{L-L_b}{b}
=\cos\be(t_k)(\theta(t_{k+1})-\theta(t_k))\,. 
\label{eq:inf-bk}
\end{align}

\ms 

Third, we obtain 
the total geometric phase $\De_g$ by summing up 
these infinitesimal angles 
for all $k$\,. 
To be precise, 
define $\be^{\text{out}}_k\,, \be^{\text{in}}_k$ 
for $0\leq k \leq N-1$ by 
\begin{align}
\cos \be^{\text{out}}_k&=
\begin{cases}
{\rm max}\{ ~\cos\be(t)~|~t_k\leq t \leq t_{k+1}~\}  & \text{if}\quad \theta(t_{k})\leq \theta(t_{k+1})\\
{\rm min}\{ ~\cos\be(t)~|~t_k\leq t \leq t_{k+1}~\}  &  \text{if}\quad\theta(t_{k})> \theta(t_{k+1})
\end{cases}
\,, 
\nln
\cos \be^{\text{in}}_k&=
\begin{cases}
{\rm min}\{ ~\cos\be(t)~|~t_k\leq t \leq t_{k+1}~\}  & \text{if}\quad \theta(t_{k})\leq \theta(t_{k+1})\\
{\rm max}\{ ~\cos\be(t)~|~t_k\leq t \leq t_{k+1}~\}  &  \text{if}\quad\theta(t_{k})> \theta(t_{k+1})
\end{cases}\,.  
\end{align}
Let $\De_g^{(k)}$ be the infinitesimal geometric phase 
for the time segment $[t_k, t_{k+1}]$\,. 
By definition, the following holds, 
\begin{align}
\cos\be^{\text{in}}_k\, (\theta(t_{k+1})-\theta(t_k))
\leq 
\De_g^{(k)}
\leq 
\cos\be^{\text{out}}_k\, (\theta(t_{k+1})-\theta(t_k))\,. 
\end{align}
Summing up for $k$\,, we get 
\begin{align}
\De_{g, N}^{\text{in}} \leq \De_{g} \leq \De_{g, N}^{\text{out}}\,,  
\end{align}
where we have denoted 
\begin{align}
\De_{g, N}^{\text{in}}&=\sum_{k=0}^{N-1}
\cos\be^{\text{in}}_k\, (\theta(t_{k+1})-\theta(t_k)) \,, 
\qquad 
\De_{g}=\sum_{k=0}^{N-1}\De_g^{(k)}\,, 
\nln
\De_{g, N}^{\text{out}}&=\sum_{k=0}^{N-1}
\cos\be_k^{\text{out}}\, (\theta(t_{k+1})-\theta(t_k))\,. 
\end{align}
Both $\De_{g, N}^{\text{in}}$ and $\De_{g, N}^{\text{out}}$ are the bounded monotone
sequences with respect to $N$\,, and they converge to the same
value. 
See \figref{fig:De_inout}\,. 
\begin{figure}
\centering
\begin{subfigure}{0.3\columnwidth}
\centering
\includegraphics[width=\columnwidth]{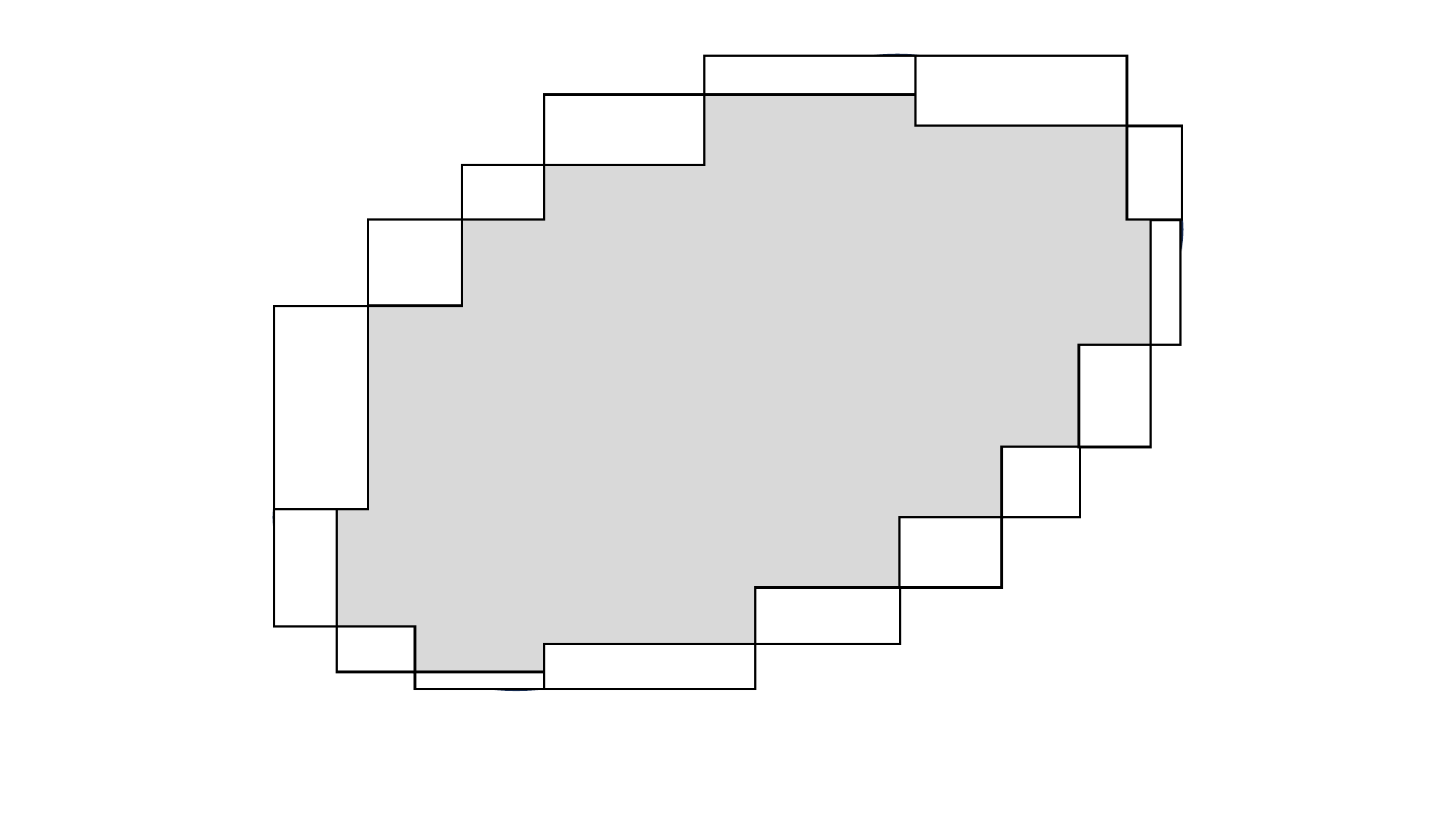}
\caption{Area measured by $\De_{g,N}^{\rm in}$\,.
}
\label{fig:De_int}
\end{subfigure}
\hspace{3mm}
\begin{subfigure}{0.3\columnwidth}
\centering
\includegraphics[width=\columnwidth]{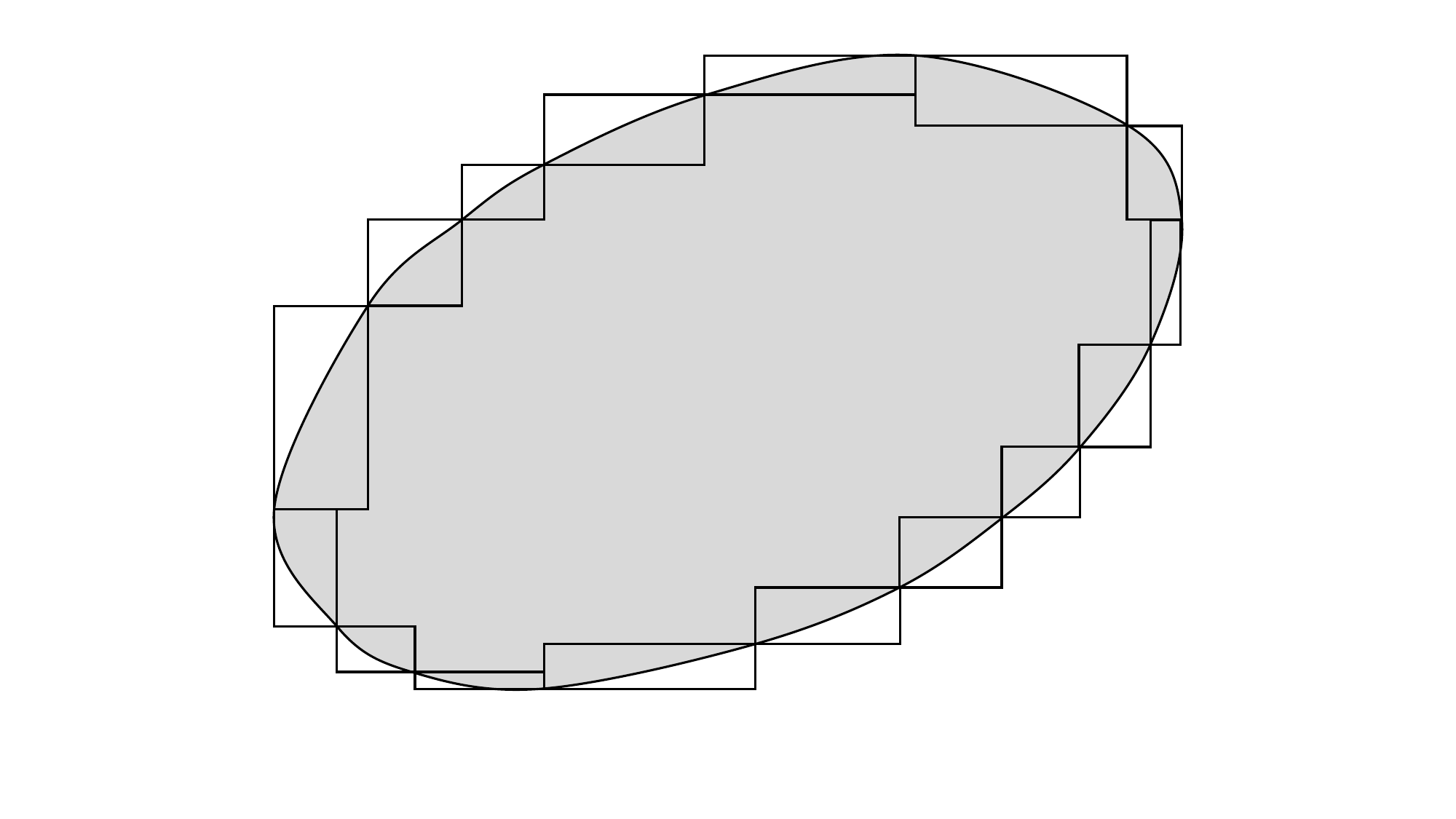}
\caption{Area measured by $\De_{g}$\,.
}
\label{fig:De_ture}
\end{subfigure}
\hspace{3mm}
\begin{subfigure}{0.3\columnwidth}
\centering
\includegraphics[width=\columnwidth]{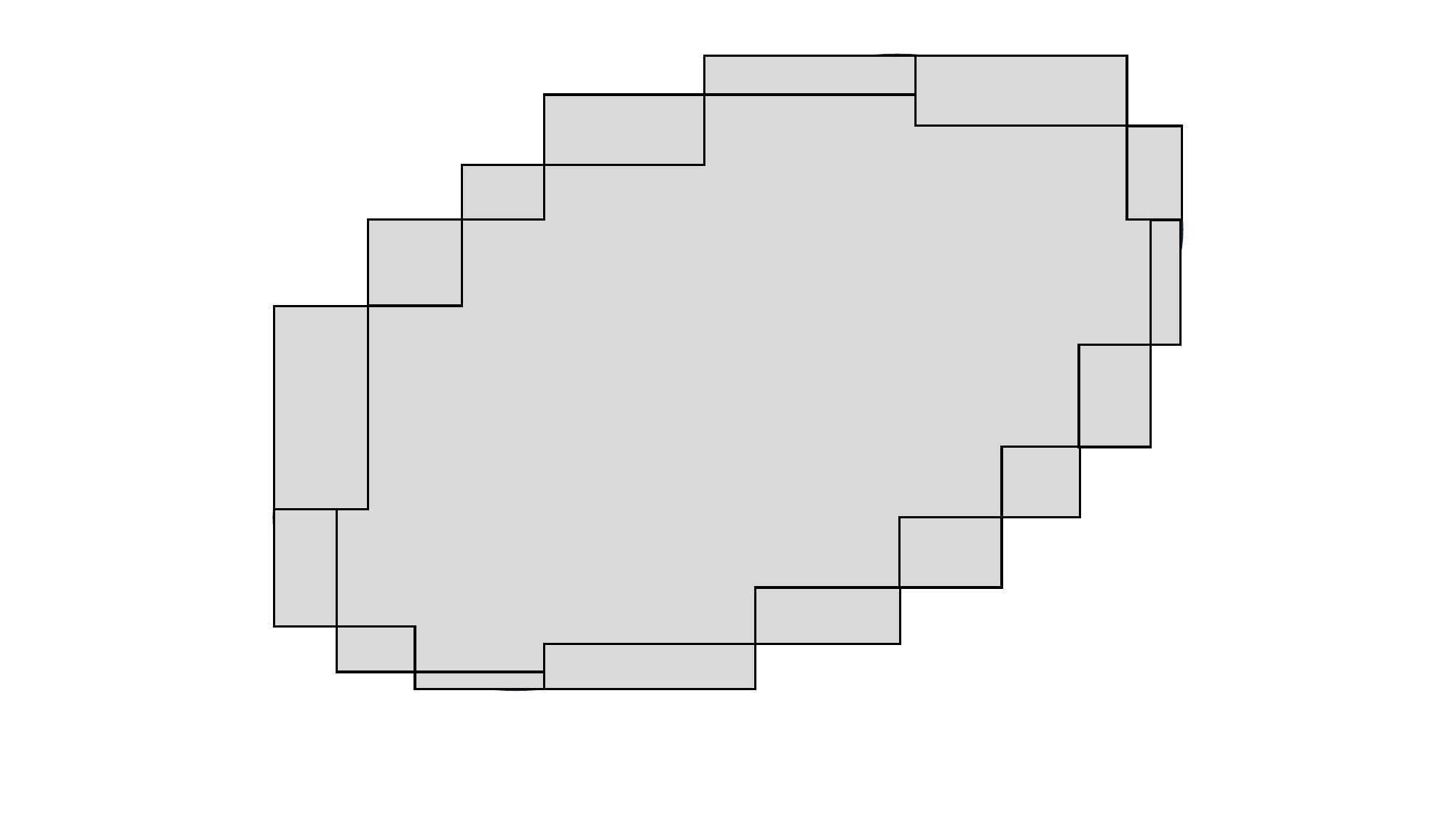}
\caption{Area measured by $\De_{g,N}^{\rm out}$\,.
}
\label{fig:De_out}
\end{subfigure}
\caption{The areas measured by 
$\De_{g,N}^{\rm in}\,, \De_{g}\,,\De_{g,N}^{\rm out}$
are indicated in grey in Figures (a), (b), (c), respectively. 
The horizontal direction is $\theta$ and the vertical is $\cos\be$\,. 
The sign of $\De_g$ is positive if the boundary is oriented clockwise. 
}
\label{fig:De_inout}
\end{figure}
In fact, it is evaluated as 
\begin{align}
\left|\De_{g, N}^{\text{out}}-\De_{g, N}^{\text{in}}\right|
&\leq  
\sum_{k=0}^{N-1}
\bigl|\cos\be_k^{\text{out}}-\cos\be_k^{\text{in}}\bigr|
\bigl|\theta(t_{k+1})-\theta(t_k)\bigr| 
\nln
&\leq 
\sum_{k=0}^{N-1}
\bigl|\be_k^{\text{out}}-\be_k^{\text{in}}\bigr|
\bigl|\theta(t_{k+1})-\theta(t_k)\bigr|\,.  
\end{align}
Since $\be(t)$ is continuous, for any $\ep>0$\,, 
there exists an enough large $N$ such that 
\begin{align}
\bigl|\be_k^{\text{out}}-\be_k^{\text{in}}\bigr|<\ep \,. 
\end{align} 
In addition, 
due to the Lipschitz conditions for $\theta(t)$ 
in \remref{rem:lip}\,, 
there exists a positive constant $C_\theta$ such that 
\begin{align}
\bigl|\theta(t_{k+1})-\theta(t_k)\bigr|\leq 
C_\theta(t_{k+1}-t_k)=\frac{C_\theta}{N}\,. 
\end{align}
Then, we have  
\begin{align}
\left|\De_{g, N}^{\text{out}}-\De_{g, N}^{\text{in}}\right|
\leq \frac{\ep C_\theta}{N}\,,  
\end{align}
which implies that 
\begin{align}
\left|\De_{g, N}^{\text{out}}-\De_{g, N}^{\text{in}}\right|
\to 0 \qquad (N\to \infty)\,. 
\end{align}
Hence, by the squeezing lemma, we get 
\begin{align}
\De_g 
=\lim_{N\to \infty} \De_{g, N}^{\text{out}}
=\lim_{N\to \infty} \De_{g, N}^{\text{in}}
\,. 
\end{align}
This proves that the geometric phase $\De_g$ is 
equal to the line integral in \eqref{eq:thm-line-int}\,. 
\end{proof}

\ms 

Though the above theorem gives us an explicit formula
for the geometric phase, the geometric meaning is opaque. 
We shall elaborate on this point in the following subsections.  

\ms 

\begin{remark}
Fig.\,\ref{fig:inf_bk_cf} 
intuitively explains the reason why the total rotation angle $\De$ of disc B is given by 
the sum of the dynamical phase $\De_d$ and the geometric phase $\De_g$ as 
claimed in \eqref{eq:dandg}\,. 
Opening up the cone along the line $OP$ in Fig.\,\ref{fig:inf_bk_cone}\,, 
we get a fan in Fig.\,\ref{fig:inf_bk_flat}\,. 
In Fig.\,\ref{fig:inf_bk_flat}\,, the total rotation angle can be decomposed into  
\begin{align}
\De=\angle CB'Q=\angle CB'P'+\angle P'B'Q\,. 
\end{align}
Since the two line $BO$ and $CB'$ are parallel, 
the angle $\angle CB'P'$ is identified with 
the geometric phase,  
\begin{align}
\De_g=\angle CB'P'=\angle POP'=\theta_B\,. 
\end{align}
On the other hand,  the angle $\angle P'B'Q$
corresponds to the dynamical phase because of 
\begin{align}
\De_d=\angle P'B'Q =\frac{\arc{P'Q}}{b}
=\frac{\arc{P'P}}{b}\,. 
\end{align}
\end{remark}

\subsection{Gauss map and the regularized curve}
\label{subsec:gmap}
\subsection*{Gauss map}

First, we introduce the unit normal vector of disc B by {\it the Gauss map}, 
which detects the motion of disc B. 
We assume that two-sphere $S^2$ is 
embedded in $\RR^3$ as 
\begin{align}
S^2=\{~x^2+y^2+z^2=1~|~\begin{pmatrix} x\\ y\\ z  \end{pmatrix}
\in \RR^3~\}\subset \RR^3 \,. 
\notag 
\end{align}
\begin{definition}[Gauss map]
\label{def:gmap}
For the two parameters $(\theta, \beta)$ in the motion 
\eqref{eq:motion}\,, 
define {\it the Gauss map} from $\RR \times [0,\pi]$
to two-sphere $S^2$ by\footnote{
The range of the angle variable $\be$ is restricted in $[0,\pi]$ 
in our paper. 
If we relax this condition from $[0,\pi]$ to $[0,2\pi]$\,, 
we may need to consider 
a torus $T^2$ rather than a sphere $S^2$ as the target space of the Gauss map.  
The geometric phases on a torus are investigated in \cite{Gho}.} 
\begin{align}
\RR \times [0,\pi]\to S^2\subset \RR^3
\,, \quad 
(\theta, \beta)
\mapsto 
\bg(\theta, \beta)=\begin{pmatrix}
\sin\be \cos\theta \\ 
\sin\be \sin\theta \\ 
-\cos\be
\end{pmatrix} \,. 
\label{eq:gauss_a}
\end{align}
The vector $\bg(\theta, \beta)$ is called {\it the Gauss vector}. 
\end{definition}

The Gauss vector is geometrically realized both in the rotating model
and two-sphere as in \figref{fig:Gmap-1} and 
\figref{fig:Gmap-2}\,, respectively. 

\ms 

\begin{figure}
\centering
\begin{subfigure}{0.45\columnwidth}
\centering
\includegraphics[width=1.3\columnwidth]{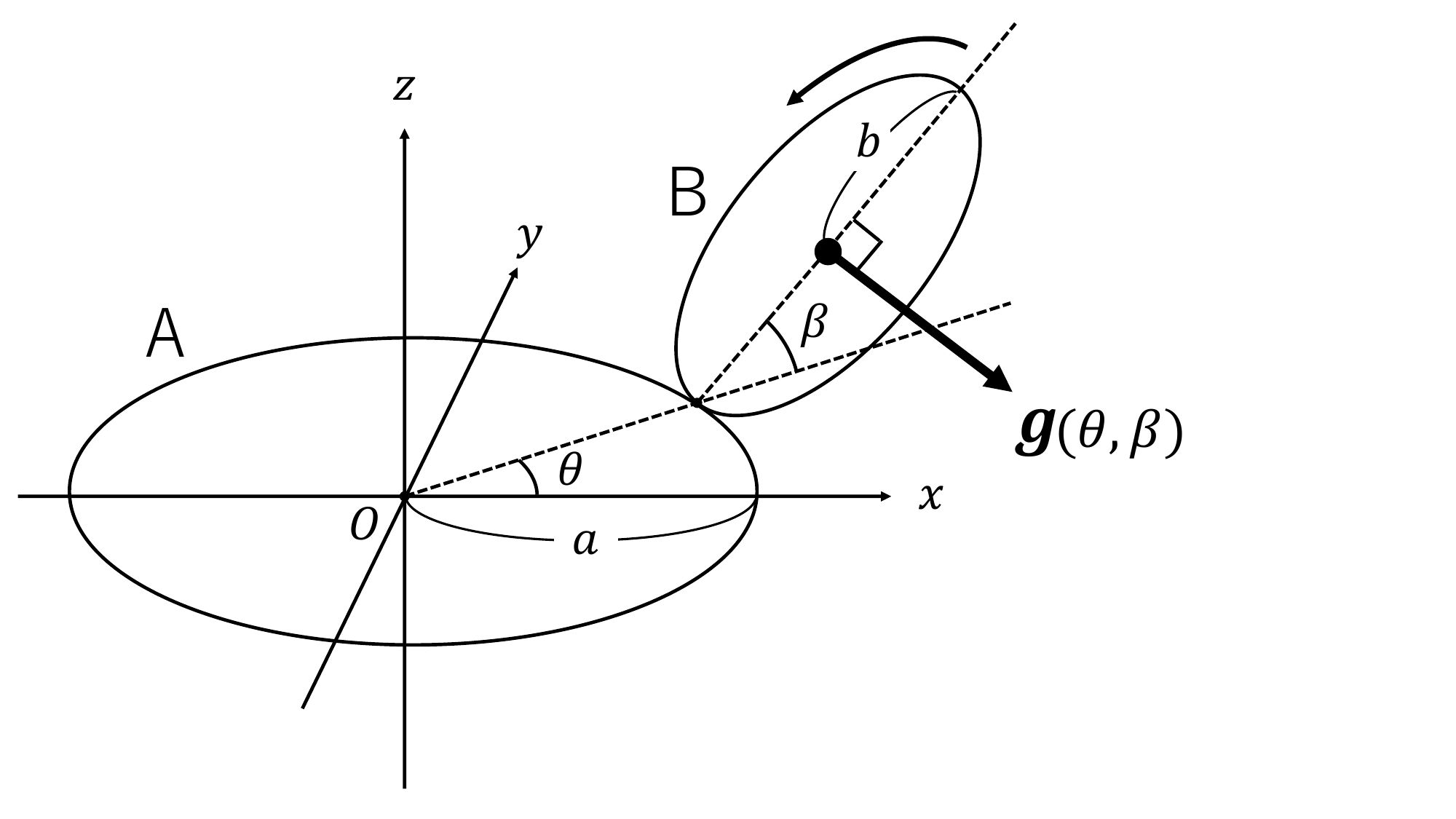}
\caption{The Gauss vector  $\bg$  in the model}
\label{fig:Gmap-1}
\end{subfigure}
\hspace{3mm}
\begin{subfigure}{0.45\columnwidth}
\centering
\includegraphics[width=1.4\columnwidth]{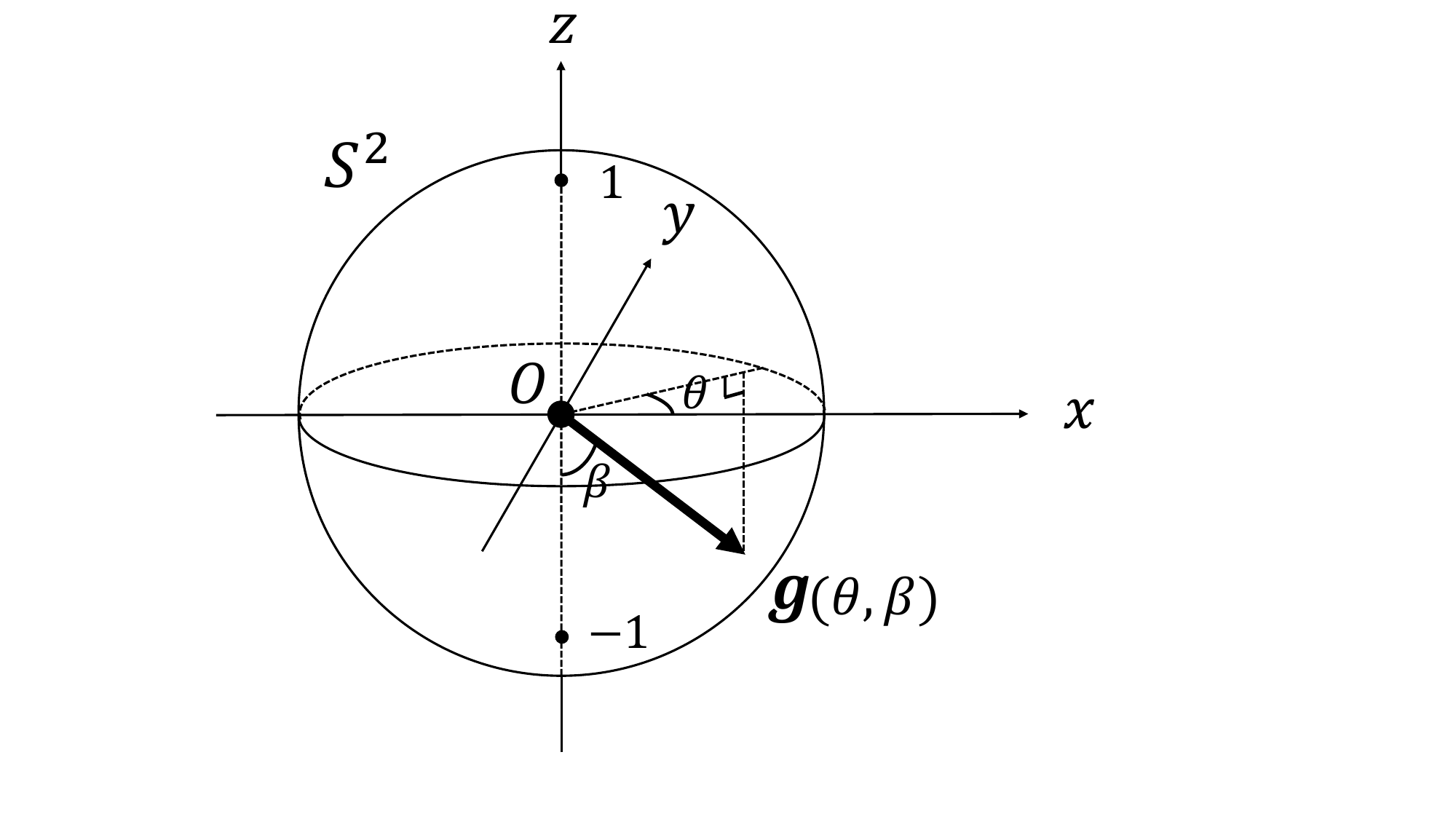}
\caption{The Gauss vector $\bg$ in two-sphere}
\label{fig:Gmap-2}
\end{subfigure}
\caption{The Gauss map $\bg$ from the model (a) 
to two-sphere (b)}
\label{fig:Gmap}
\end{figure}

Note that $\bg(\theta, \beta)$ is a unit vector and 
perpendicular to disc B.    
By composing the motion \eqref{eq:motion} and the Gauss
map \eqref{eq:gauss_a}, the Gauss vector could be regarded
as a function of $t\in [0,1]$\,, 
\begin{align}
\bg : [0,1] \to S^2\,, \quad t\mapsto \bg(t)
=\begin{pmatrix}
\sin\be(t) \cos\theta(t) \\ 
\sin\be(t) \sin\theta(t) \\ 
-\cos\be(t)
\end{pmatrix} \,. 
\label{eq:gauss}
\end{align}
We also refer to the map \eqref{eq:gauss} as the Gauss map. 

\ms 

Because of $\bg(1)=\bg(0)$\,, the Gauss map defines 
the oriented closed curve $\ga$ on $S^2$ 
associated with 
the given motion \eqref{eq:motion} by 
\begin{align}
\ga=\{~\bg(t)\in S^2~|~t\in [0,1]~\}\subset S^2\,.  
\label{eq:ga}
\end{align}
The orientation of $\ga$ is induced from that of 
the parameter $t\in [0,1]$ such as from $t=0$ to $t=1$\,.  
The curve $\ga$ is said to be {\it simple} if it does not intersect with itself.

\subsection*{The regularized curve}

Second, let us introduce the {\it cut-off} with a sufficiently small $\ep$ 
$(0<\ep<\pi/8)$ by 
\begin{align}
&\be_\ep: [0,1] \to [\ep \,, \pi-\ep ] \,, 
\qquad t\mapsto \be_\ep(t) \,, 
\nln 
&\be_\ep(t)=\begin{cases}
\ep & (~0\leq \be(t) \leq \ep~) \\
\be(t) &(~\ep < \be(t) <\pi-\ep~) \\ 
\pi-\ep & (~\pi- \ep \leq \be(t) \leq \pi~) \,. 
\end{cases}
\end{align}
The regularized Gauss vector is then given by replacing $\be(t)$
by $\be_\ep(t)$ as  
\begin{align}
\bg_\ep : [0,1] \to S^2-\{0,0,\pm1\}\,, \quad 
t\mapsto \bg_\ep(t)
=\begin{pmatrix}
\sin\be_\ep(t) \cos\theta(t) \\ 
\sin\be_\ep(t) \sin\theta(t) \\ 
-\cos\be_\ep(t)
\end{pmatrix} \,. 
\label{eq:gauss-ep}
\end{align}

\ms 

\begin{definition}[Regularized curve]
\label{def:ga-reg}
The {\it regularized curve} is defined by the orbit
of the regularized Gauss vector,  
\begin{align}
\ga(\ep)=\{~\bg_\ep(t)\in S^2-\{(0,0,\pm1)\}~
|~t\in [0,1]~\}\subset 
S^2-\{(0,0,\pm1)\}\,.  
\label{eq:ga-reg}
\end{align}
\end{definition}

\ms 

The point is  that the curve $\ga(\ep)$ dodges the poles
$(0,0,\pm1)$\,, though the undeformed curve 
$\ga$ does not in general.  
By taking $\ep \to 0$\,, it reduces to the original path, 
\begin{align}
\ga(\ep) \to \ga\,. 
\end{align}

\ms 

\begin{prop}
\label{prop:gp-reg}
The geometric phase for the motion \eqref{eq:motion}
is given by the limit of the line integral along the regularized curve
$\ga(\ep)$ in \eqref{eq:ga-reg}\,, 
\begin{align}
\De_g
=\lim_{\ep\to0}\int_{\ga(\ep)} \cos\be_\ep d\theta
\label{eq:gp-reg}
\end{align}
\end{prop}

\ms 

\begin{proof}
The right hand side is calculated as 
\begin{align}
\lim_{\ep\to0}\int_{\ga(\ep)} \cos\be_\ep d\theta
=\lim_{\ep\to0}\int_{0}^{1} \cos\be_\ep(t)
\frac{d\theta(t)}{dt} dt  
=\int_{0}^{1} \cos\be(t)
\frac{d\theta(t)}{dt} dt\,.   
\end{align}
By \dthmref{thm:line-int}\,, this coincides with the 
geometric phase $\De_g$\,. 
\end{proof}

\begin{remark}
\label{rem:ga12}
The regularized curve  $\ga(\ep)$ always and uniquely  
exists for the motion \eqref{eq:motion}\,.  
Note that, in some cases, 
two mutually different curves $\ga_1(\ep)\,, \ga_2(\ep)$ give rise to 
the same curve $\ga$ in $\ep\to0$ limit
{\it i.e.}\,, 
\begin{align}
\lim_{\ep\to0} \ga_1(\ep)=
\lim_{\ep\to0} \ga_2(\ep)=\ga\,. 
\end{align}
However, the resulting geometric phases do {\it not} coincide in general,  
\begin{align}
\lim_{\ep\to0}\int_{\ga_1(\ep)} \cos\be d\theta
\neq
\lim_{\ep\to0}\int_{\ga_2(\ep)} \cos\be d\theta\,. 
\end{align}
This is because the Gauss vector $\bg$ cannot detect the motion
when it is at the poles $(0,0,\pm1)$\,. 
In this sense, the regularized Gauss vector $\bg_\ep$ in \eqref{eq:gauss-ep}
more faithfully describes the motion of disc B. 
%
In other words, 
the geometric phase $\De_g$ is a functional 
of curve $\ga$\,, but it is {\it not} continuous for the 
infinitesimal variation of $\ga$\,.  
Indeed, the discontinuity by $\pm2\pi n\; (n\in \ZZ)$  
occurs when the curve $\ga$ jumps
over the poles at $(0,0,\pm1)$\,. 
Examples (v) and (iv) in  \secref{subsec:ex} well demonstrate 
what we have mentioned here.  
\end{remark}

\subsection{Connection matrix and the geodesic curvature}
\label{subsec:w12}

Next, we shall elucidate the geometric meaning of the one-form
$\cos\be d\theta$ in \eqref{eq:gp-reg}\,.

\subsection*{Connection matrix}

Using the Gauss vector \eqref{eq:gauss_a}\,, 
we introduce the local orthonormal frame of $\RR^3$
at $\bg(\theta, \be)\in S^2$ by 
\begin{align}
& \ee_1=\frac{\partial \bg}{\partial\theta} \Big/ 
\Big|\frac{\partial \bg}{\partial\theta}\Big|
=\begin{pmatrix} -\sin \theta \\ \cos \theta \\ 0  \end{pmatrix}\,, 
\quad 
\ee_2=\frac{\partial \bg}{\partial\be} \Big/ 
\Big|\frac{\partial \bg}{\partial\be}\Big|
=\begin{pmatrix} 
\cos\be \cos \theta \\ \cos\be \sin \theta \\ \sin\be  \end{pmatrix}\,, 
\notag \\
&
\ee_3=\ee_1\times \ee_2 
=\begin{pmatrix}\sin\be \cos\theta \\ \sin\be \sin\theta \\ -\cos\be \end{pmatrix}\,.  
\label{eq:e123}
\end{align}
It is noticed that $\ee_3=\bg$ and they satisfy 
\begin{align}
&\ee_i \cdot \ee_j=\de_{ij}\qquad (i,j=1,2,3)\,, 
\label{eq:ortho}
\\
&\ee_1 \times \ee_2=-\ee_2 \times \ee_1=\ee_3 \,, 
\quad 
\ee_2 \times \ee_3=-\ee_3 \times \ee_2=\ee_1 \,, 
\quad 
\ee_3 \times \ee_1=-\ee_1 \times \ee_3=\ee_2 \,. 
\notag 
\end{align}
Noting that the vectors $\ee_1$ and $\ee_2$ are always oriented to 
the {\it east} and {\it north}, 
respectively, and $\ee_3$ is to the {\it sky}.  
In particular, $\ee_1$ and $\ee_2$ span the tangent space
$T_{\bg}S^2$\,. 

\ms 

The {\it connection matrix} $(\om_{ij})$ is defined by 
\begin{align}
\begin{pmatrix} d\ee_1 \\ d\ee_2 \\ d\ee_3   \end{pmatrix}
=
\begin{pmatrix}
\om_{11} & \om_{12} & \om_{13} \\ 
\om_{21} & \om_{22} & \om_{23} \\ 
\om_{31} & \om_{32} & \om_{33} 
\end{pmatrix}
\begin{pmatrix} \ee_1 \\ \ee_2 \\ \ee_3   \end{pmatrix}\,. 
\end{align}
By the orthonormal condition \eqref{eq:ortho}\,, it immediately follows that 
$\om$ is anti-symmetric 
\begin{align}
\om_{ij}=-\om_{ji}=d\ee_i\cdot \ee_j
\qquad \text{for}\qquad 
i,j=1,2,3\,. 
\end{align}
It is explicitly calculated as 
\begin{align}
\begin{pmatrix}
\om_{11} & \om_{12} & \om_{13} \\ 
\om_{21} & \om_{22} & \om_{23} \\ 
\om_{31} & \om_{32} & \om_{33} 
\end{pmatrix}
=
\begin{pmatrix}
0 & -\cos\be d\theta & -\sin\be d\theta \\ 
\cos\be d\theta & 0 & -d\be \\ 
\sin\be d\theta & d\be & 0 
\end{pmatrix}\,. 
\label{eq:cmat}
\end{align}
In particular, we see that $\om_{21}=\cos\be d\theta$\,. 
Hence, by \propref{prop:gp-reg}\,, 
we have the following proposition. 
\begin{prop}
It holds that 
\begin{align}
\De_g=\lim_{\ep\to0} \int_{\ga(\ep)} \om_{21}\,. 
\end{align}
\end{prop}

\subsection*{Reparameterization from $t$ to $s$}
To describe the motion of disc B, it is convenient to adopt 
the {\it length} of the curve $\ga$ rather than the {\it time} parameter $t\in [0,1]$\,. 
We shall express the length parameter by $s\in [0,L(\ga)]$\,,
where $L(\ga)$ is the length of $\ga\subset S^2$\,, 
and regard $\bg(s)$ as a function of the length $s$ 
rather than the time $t$\,.
The orientation of the curve parametrized by $s$ is induced from 
that of the time parameter $t$ such that    
\begin{align}
\bg(t)\big|_{t=0}=\bg(s)\big|_{s=0}
\qquad \text{and}\qquad 
\bg(t)\big|_{t=1}=\bg(s)\big|_{s=L(\ga)}\,.
\end{align}

\ms 

On an open neighborhood for a fixed $t\in (0,1)$\,, 
where $\bg(t)$ is differentiable and $\bg'(t)\neq0$\,, 
the length parameter $s$ is related to 
the time parameter $t$ via the reparameterization
\begin{align}
s : [0,1] \to [0, L(\ga)]\,, \quad t\mapsto s(t) 
\qquad \text{such that }\qquad 
\left|\bg'(s)\right|= \left|\bg'(t)\right|
\frac{dt}{ds}=1\,.
\label{eq:repara}
\end{align}
By the definition, the speed of $\bg(s)$ is normalized 
as $1$ for the parameter $s$\,. 
We also see that $ds^2$ is the line element of $S^2$. By \eqref{eq:repara} and 
\eqref{eq:cmat}\,, we have 
\begin{align}
ds=\left|\bg'(t)\right| dt\,, 
\qquad 
\bg'(t)=\frac{d\ee_3}{dt}=\sin\be \frac{d\theta}{dt}\ee_1+\frac{d\be}{dt}\ee_2\,.  
\end{align}
Therefore, we obtain
\begin{align}
ds^2=\sin^2\be d\theta^2+d\be^2\,. 
\end{align}
This is nothing but the line element of $S^2$ in terms of the local coordinate $(\theta, \be)$\,. 
In the subsequent argument, we reserve $s$ for the length parameter.

\subsection*{Geodesic curvature}

Due to $|\bg'(s)|=1$ and $\bg'(s)\in T_{\bg(s)}S^2
=\text{span}\{\ee_1\,,\ee_2\}$\,, 
let us express the tangent vector $\bg'(s)$ by using an angle 
$\vp$\,, 
\begin{align}
\bg'(s)=\cos\vp\, \ee_1+\sin\vp\, \ee_2\,, \quad \text{where}\quad 
\cos\vp=\sin \be \frac{d\theta}{ds}\,, \quad 
\sin\vp=\frac{d\be}{ds}\,.  
\label{eq:g-prime}
\end{align}
See \figref{fig:frame}. 
Replacing $\vp$ by $\vp+\pi/2$\,, we have the orthogonal vector defined by 
\begin{align}
\bnu(s)=-\sin\vp\, \ee_1+\cos\vp\, \ee_2\,. 
\label{eq:nu}
\end{align}
Note that they satisfy 
\begin{align}
|\bg'(s)|=|\bnu(s)|=1 \,,\qquad  \bg'(s)\cdot \bnu(s)=0\,. 
\label{eq:gvgv-ortho}
\end{align}
The relation between the local frames $\{\ee_1\,,\ee_2\}$
and $\{\bg'(s)\,, \bnu(s)\}$ is presented in 
\figref{fig:frame}. 
\begin{figure}
\centering
\includegraphics[width=10cm]{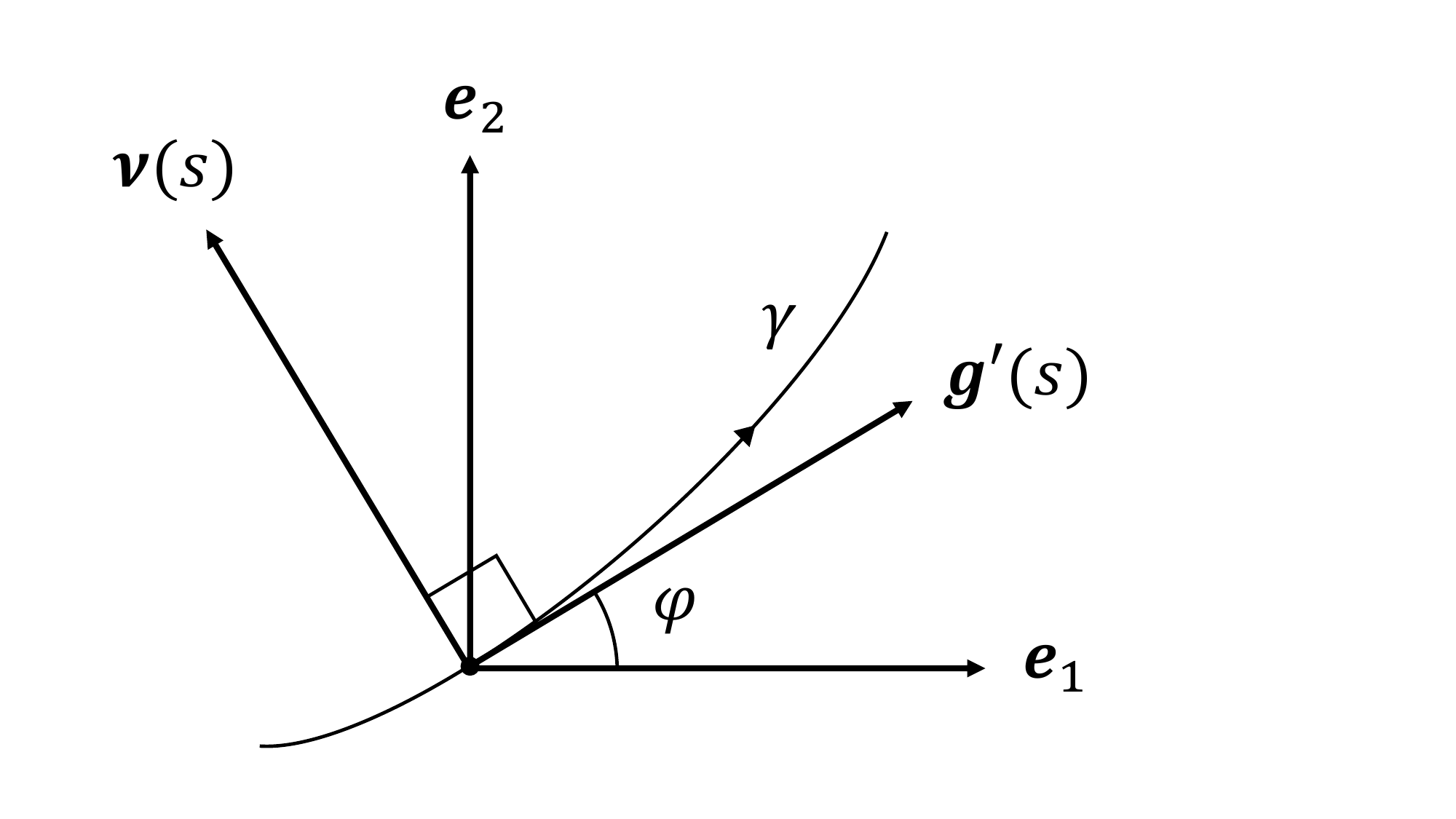}
\caption{The moving frames $\{\bg'(s)\,, \bnu(s)\}$ and 
$\{\ee_1\,,\ee_2\}$ at $\bg(s)\in S^2$\,. 
The vectors $\ee_1$ and $\ee_2$ are oriented to 
the east and north, respectively. 
$\vp$ is the angle between the local frames.}
\label{fig:frame}
\end{figure}
Differentiating the second relation, we get 
\begin{align}
\bg''(s)\cdot \bnu(s)+\bg'(s)\cdot \bnu'(s)=0\,. 
\label{eq:gvgv}
\end{align}

\begin{definition}[Geodesic curvature]
\label{def:curv-def}
The {\it geodesic curvature} of $\ga\subset S^2$ at $\bg(s)\in \ga$ is defined by 
\begin{align}
\ka_g(s)= \bg''(s)\cdot \bnu(s)=-\bg'(s)\cdot \bnu'(s)\,.
\label{eq:curv-def}
\end{align}
\end{definition}

Now, let us calculate the geodesic curvature explicitly. 
Plugging \eqref{eq:g-prime} and 
\begin{align}
\bnu'(s)ds 
&= -\cos \vp ( d\vp + \om_{12} )\, \ee_1 
-\sin \vp ( d\vp + \om_{12} )\, \ee_2 
+(-\sin\vp\, \om_{13}+\cos\vp\, \om_{23})\ee_3\,, 
\end{align}
with the geodesic curvature \eqref{eq:curv-def}\,, 
we obtain 
\begin{align}
\ka_g(s) ds = d\vp + \om_{12} \,. 
\end{align}

The above argument can be summarized in the following proposition. 
\begin{prop}
\label{prop:oneform}
The infinitesimal geometric phase $\cos\be d\theta$ 
in \eqref{eq:gp-reg}\,, 
the element $\om_{21}$ of the connection matrix 
\eqref{eq:cmat}\,, 
the angle $\vp$ in \eqref{eq:g-prime}\,, 
and the geometric curvature $\ka_g(s)ds$ in \eqref{eq:curv-def}
satisfy the relation 
\begin{align}
\cos\be d \theta 
=\om_{21}=d\vp-\ka_g(s) ds\,. 
\end{align}
\end{prop}

\subsection{Main theorem}
\label{subsec:thm}

By \propref{prop:gp-reg}\,, 
the geometric phase is obtained by evaluating the line integral 
along the regularized curve $\ga(\ep)$ and taking the regularization parameter 
$\ep$ to zero. 
It is, however, noticed that the curve $\ga(\ep)$ is piecewise 
smooth in general and could have finite cusp points. 
Let us suppose that the curve $\ga(\ep)$ has 
$n(\ep)\in \ZZ_{\geq0}$ cusp points
and, at the $i$-th cusp point, denote the external angle to the counterclockwise by $\al_i(\ep)$\,.  
Define 
the geodesic curvature $\ka_g^{\ep}(s)$ for
$\ga(\ep)\subset S^2-\{(0,0,\pm1)\}$ at $\bg_\ep(s)\in \ga(\ep)$ by 
\begin{align}
\ka_g^{\ep}(s)= \bg_\ep''(s)\cdot \bnu_\ep(s)=-\bg_\ep'(s)\cdot \bnu_\ep'(s)\,, 
\end{align}
where $\bnu_\ep(s)$ is the normalized orthogonal vector to $\bg_\ep'(s)$ 
as similar in \eqref{eq:nu}\,. 

\ms 

Then, the geometric phase is given by 
\begin{align}
\De_g
=\lim_{\ep\to0} \int_{\ga(\ep)} d\vp 
- \lim_{\ep\to0} \left(\int_{\ga(\ep)}\ka_g^{\ep} ds 
+\sum_{i=1}^{n(\ep)} \al_i(\ep) \right)\,. 
\end{align}
Since the second term reduces to the line integral on the 
undeformed curve $\ga$\,, 
we can express the geometric phase as follows. 
\begin{thm}
\label{thm:De-twoint}
The geometric phase for the motion \eqref{eq:motion}
is evaluated by 
\begin{align}
\De_g
=\lim_{\ep\to0} \int_{\ga(\ep)} d\vp 
- \int_{\ga}\ka_g ds 
- \sum_{i=1}^n \al_i \,. 
\label{eq:De-twoint}
\end{align}
Here, $\ga$ in \eqref{eq:ga}is the trajectory of the Gauss vector,
$\ga(\ep)$ introduced in \eqref{eq:ga-reg} 
is the regularized curve of $\ga$  avoiding the poles $(0,0,\pm1)$\,, 
$\vp$ defined in \eqref{eq:g-prime} 
is the angle of the tangent vector of $\ga$ 
measured from the local orthogonal frame,    
$\ka_g$ defined in \eqref{eq:curv-def} is 
the geodesic curvature of $\ga$\,, 
and $\al_i$ is the external angle to counterclockwise at the $i$-th cusp point of $\ga$\,. 
\end{thm}

\ms 

\begin{remark}
Theorem \ref{thm:De-twoint} is valid not only 
for a closed curve $\ga$\,, but also for some open paths
{\it not} satisfying the topological condition \eqref{eq:top}\,. 
However, when $\ga$ is closed, an interesting geometric 
interpretation exists.  
That is the claim of Theorem \ref{thm:main}. 
\end{remark}

\ms 

When the closed curve $\ga(\ep)$ does not intersect with itself, 
{\it i.e.}, simple, 
it defines two region $S_+(\ep)$ and $S_-(\ep)$ on $S^2$ such as 
\begin{align}
S^2=S_+(\ep)\cup S_-(\ep)\,, \quad 
\ga(\ep)=S_+(\ep)\cap S_-(\ep)\,, \quad 
\partial S_+(\ep)=-\partial S_-(\ep)=\ga(\ep)\,. 
\end{align}
Note that $S_+(\ep)$ is the region surrounded by 
the oriented curve $\ga(\ep)$ on the left hand side.  
Let us call $(0,0,1)$ and $(0,0,-1)$ the {\it north} and 
{\it south pole}, respectively.
Since the regularized curve $\ga(\ep)$ does not pass 
through the north and south poles, 
it makes sense to introduce the indices $I_\pm\in \{0,1,2\}$ as 
the number of poles included in $S_\pm(\ep)$\,. 
By definition, it holds 
\begin{align}
I_++I_-=2\,. 
\end{align}

Using the indices $I_\pm$, we can express 
the first term on the right hand side of \eqref{eq:De-twoint} as  
\begin{align}
\int_{\ga(\ep)} d\vp 
=-\pi(I_+-I_-)
=
\begin{cases}
2\pi  & (\,(I_+, I_-)=(0,2)\,) \\ 
0 & (\,(I_+, I_-)=(1,1)\,) \\ 
-2\pi  & (\,(I_+, I_-)=(2,0)\,) 
\end{cases}\,. 
\end{align}
Since the results are stable under the deformation 
of parameter $\ep$\,,  we have 
\begin{lemma}
\label{lem:vp}
\begin{align}
\lim_{\ep\to0} \int_{\ga(\ep)} d\vp 
=-2\pi I_++2\pi=2\pi I_--2\pi
=-\pi(I_+-I_-)\,. 
\end{align}
\end{lemma}

\ms 

Next, let us consider the second and third terms 
on the right hand side in 
\eqref{eq:De-twoint}\,. 
Denote the areas of $S_\pm(\ep)$ by $A_\pm(\ep)$\,, respectively, 
and set $S_\pm(0)=S_\pm\,, A_\pm(0)=A_\pm$\,. 
Since the area of two-sphere is $4\pi$\,, it holds that 
\begin{align}
A_+(\ep)+A_-(\ep)=4\pi. 
\end{align}
With these notations, applying the {\it Gauss-Bonnet theorem} for 
$S_+(\ep)$\,, we obtain 
\begin{align}
\int_{S_+(\ep)} K dA
+ \int_{\partial S_+(\ep)} \ka_g ds
+\sum_{i=1}^{n(\ep)} \al_i(\ep) 
=2\pi\chi(S_+(\ep))\,,  
\end{align}
where $K$ is the {\it Gauss curvature}, 
$dA=\sin\be\,  d\theta \wedge d\be$
is the volume form of $S^2$\,,  
and $\chi(S_+(\ep))$ is the {\it Euler characteristic} 
of $S_+(\ep)$ which is topologically 
isomorphic to a closed unit disc. 
Since $K=1$ for two-sphere,
the first term on the left hand side is equal to  the area of 
$S_+(\ep)$\,, 
\begin{align}
\int_{S_+(\ep)} K dA
=\int_{S_+(\ep)} \sin\be\, d\theta \wedge d\be
=A_+(\ep) \,. 
\end{align}
Taking into account $\chi(S_+(\ep))=1$\,, 
we have 
\begin{align}
-\int_{\partial S_+(\ep)} \ka_g ds
-\sum_{i=1}^{n(\ep)} \al_i(\ep) 
=A_+(\ep)-2\pi\,. 
\end{align}
By sending $\ep\to0$\,, we obtain 
\begin{lemma}
\label{lem:GB}
\begin{align}
-\int_{\ga} \ka_g ds  -\sum_{i=1}^n \al_i 
= A_+-2\pi=-A_- + 2\pi 
=\frac{A_+-A_-}{2} \,. 
\end{align}
\end{lemma}



%
%

We are now ready to state our main theorem, 
which solves the Fundamental Problem. 
\begin{thm}[Main Theorem]
\label{thm:main}
Suppose that the oriented closed curve $\ga$ defined in 
\eqref{eq:ga} is simple.
Then, for the motion \eqref{eq:motion}\,, 
the geometric phase is given by 
\begin{align}
\De_g
&=A_+-2\pi I_+
=-A_-+2\pi I_- 
=\frac{A_+-A_-}{2}-\pi(I_+-I_-) \,. 
\label{eq:mainthm}
\end{align}
\end{thm}

\begin{figure}[h]
\centering
\includegraphics[width=15cm]{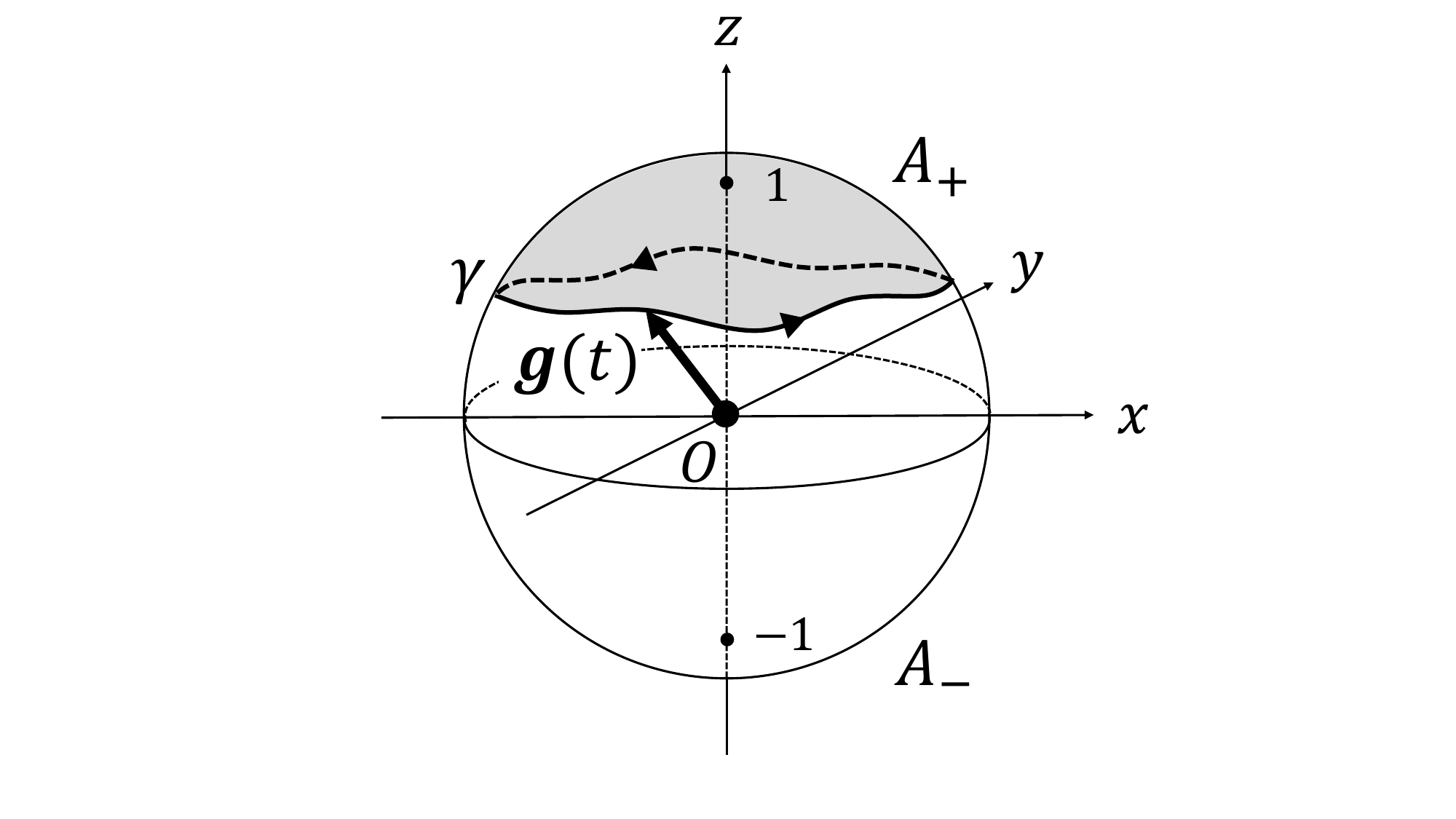}
\caption{The curve $\ga$ in $S^2$ appears as a trajectory of 
the Gauss vector $\bg(t)$ whose orientation is induced 
by the parameter $t$\,. 
$A_+$ is the area enclosed by $\ga$ on the left side, while 
$A_-$ is on the right. 
The topological index $I_+ (I_-)$ is the number of poles 
$(0,0,\pm1)$ 
contained in $\ga$ on the left (right, respectively). 
In the above picture, $I_+=I_-=1$\,. 
}
\label{fig:mainthm}
\end{figure}%

\begin{proof}	
Applying \lemref{lem:vp} and \lemref{lem:GB} for 
\thmref{thm:De-twoint}\,, 
we immediately get 
\begin{align}
\De_g
=(-2\pi I_++2\pi)+(A_+-2\pi) =A_+-2\pi I_+\,. 
\end{align}
The other expressions easily follow from 
$A_++A_-=4\pi$ and $I_++I_-=2$\,. 
\end{proof}

\ms

Our answer to the question is given below. 

\ms 

\begin{itembox}[l]{\bf Answer for the Question}
The rotation angle $\De$ of the rotating disc B in the motion 
\eqref{eq:motion} with the topological condition \eqref{eq:top}
is given by the sum of the dynamical phase $\De_d$
and the geometric phase $\De_g$ as follows, 
\begin{align}
\De=\De_d+\De_g  \quad \text{with}\quad 
\De_d=\frac{2\pi n a}{b} 
\,,\quad 
\De_g=A_+-2\pi I_+\,.
\label{eq:ans}
\end{align}
Here, $A_+$ is the area of the region surrounded by the curve 
$\ga$ in \eqref{eq:ga} which is supposed to be simple, 
and $I_+\in \{0,1,2\}$ is a number of the poles 
enclosed by the regularized curve $\ga(\ep)$ in \eqref{eq:ga-reg}
on the left side. 
\end{itembox}
\ms 

\begin{remark}
If $\ga$ is not simple, $\De_g$ is obtained by decomposing
it into the sum of simple closed curves and repeatedly applying the formula \eqref{eq:ans} for each segment. 
\end{remark}

\section{Applications for some examples}
\setcounter{equation}{0}
\label{sec:ex}

In this section, we apply our solution for the typical motions. 
After collecting the results we have obtained so far in \ssecref{subsec:results}\,, 
we will apply them to six examples in 
\ssecref{subsec:ex}\,.

\subsection{Collection of the results}
\label{subsec:results}

The total rotation angle $\De$ consists of 
the dynamical phase $\De_d$ and 
the geometric phase $\De_g$\,, 
that is $\De=\De_d+\De_g$ 
in \eqref{eq:dandg}\,. 
The dynamical phase is given by 
\begin{align}
\De_d=\frac{a\theta(1)}{b}=\frac{2\pi na}{b}\,, 
\tag{\dpropref{prop:dyn}}
\end{align}
and the geometric phase is 
\begin{align}
\De_g=A_+-2\pi I_+\,, 
\tag{\thmref{thm:main} (Main theorem)}
\end{align}
where $A_+$ is the area surrounded by the curve 
$\ga\subset S^2$ on the left side, 
and $I_+\in \{0,1,2\}$ is the number of poles enclosed by the 
regularized curve $\ga(\ep)\subset S^2-\{(0,0,\pm1)\}$\,. 
Their explicit relations are the followings, 
\begin{align}
A_+&=2\pi -\int_\ga \ka_g(s)ds -\sum^n_{i=1}\al_i \,,
\tag{Lem.\,\ref{lem:GB}}
\\
2\pi I_+&=2\pi - \lim_{\ep\to0}\int_{\ga(\ep)}d\vp \,. 
\tag{Lem.\,\ref{lem:vp}}
\end{align}
We have gained \thmref{thm:main} 
by rewriting 
\begin{align}
\De_g
=\int_{0}^{1} \cos\be(t) \frac{d\theta(t)}{dt} dt 
=\lim_{\ep\to0}\int_{\ga(\ep)} \cos\be_\ep\, d\theta
\tag{\dthmref{thm:line-int} and Prop.\,\ref{prop:gp-reg}}
\end{align}
according to \propref{prop:oneform} and  
plugging the above expressions in it.

\subsection{Examples}
\label{subsec:ex}

We apply the above results to some typical examples. 
Here, we discuss the six examples. The first four examples, 
(i), (ii), (iii), and (iv), treat the cases in which $\be$ is constant. 
While, the last two examples, (v) and (vi), focus 
on the cases in which $\be$ varies. 
All data for the rotation angles are listed in Tab.\,\ref{tab:ex}\,. 


\begin{table}[htbp]
\centering
\setlength{\tabcolsep}{5pt}
\caption{All data to find out the rotation angles 
$\De$ for the six examples in \secref{subsec:ex}\,.  
Note that $\De=\De_d+\De_g$ and $\De_g=A_+-2\pi I_+$\,. 
}
\begin{tabular}{|c|c|c|c|c|c|c|} \hline
& \dis{\De_d} & \dis{A_+} & \dis{2\pi I_+}  
& \dis{\De_g}
& \dis{\De}
& Figures
\\ \hline\hline 
\phantom{\Bigg|}(i)\phantom{\Bigg|}  
&\dis{\frac{2\pi a}{b}} 
& $4\pi$ 
& $2\pi$ 
& $2\pi$ 
& \dis{\frac{2\pi a}{b}+2\pi} 
& \figref{fig:exam-zero} \\ \hline
\phantom{\Bigg|}(ii)\phantom{\Bigg|} 
& \dis{\frac{2\pi a}{b}} 
& $2\pi$ 
& $2\pi$ 
& $0$ 
& \dis{\frac{2\pi a}{b}} 
&\figref{fig:exam-pihalf} \\ \hline
\phantom{\Bigg|}(iii)\phantom{\Bigg|} 
&\dis{\frac{2\pi a}{b}} 
& $0$ 
& $2\pi$ 
& $-2\pi$ 
& \dis{\frac{2\pi a}{b}-2\pi} 
&\figref{fig:exam-pi} \\ \hline
\phantom{\Bigg|}(iv)\phantom{\Bigg|} 
&\dis{\frac{2\pi a}{b}} 
& $2\pi(1+\cos\be_0)$ 
&$2\pi$ 
& $2\pi \cos\be_0$ 
&\dis{\frac{2\pi a}{b}+2\pi \cos\be_0} 
&\figref{fig:exam-beta} \\ \hline
\phantom{\Bigg|}(v)\phantom{\Bigg|} 
&\dis{0} 
& \dis{\frac{\pi}{2}}
&\dis{0}
&\dis{\frac{\pi}{2}}
&\dis{\frac{\pi}{2}} 
&\figref{fig:exam-triang-minus}\,, 
\ref{fig:exam-triang}\,,
\ref{fig:exam_v}
\\ \hline
\phantom{\Bigg|}(vi)\phantom{\Bigg|} 
&\dis{-\frac{2\pi a}{b}} 
& \dis{\frac{\pi}{2}} 
&\dis{2\pi} 
&\dis{-\frac{3\pi}{2} } 
&\dis{- \frac{2\pi a}{b}-\frac{3\pi}{2}}
&\figref{fig:exam-triang-plus}\,,
\ref{fig:exam-triang}\,,
\ref{fig:exam_vi}
\\ \hline
\end{tabular}
\label{tab:ex}
\end{table}%

\ms 

The first four examples consider cases where $\be$ is constant. 
In the examples (i), (ii), and (iii), we shall revisit our observations in 
\ssecref{subsubsec:obs}\,. 
Example (iv) interpolates the three examples.  

\ms 
\begin{itemize}
\item[(i)] 
The first case is the motion \eqref{eq:motion} defined by 
\begin{align}
\theta(t)=2\pi t\,, \qquad \be(t)=0\,. 
\end{align}
In this case, the dynamical phase reads \dis{\De_d=2\pi a/b\,.}
The curve $\ga$ is just a point at the {\it south pole} $(0,0,-1)$\,, 
and the regularized curve  $\ga(\ep)$ is a small circle enclosing it on the right side. 
See Fig.\,\ref{fig:exam-zero}. 
Hence, $A_+$ is the area of $S^2$\,, namely $4\pi$\,, and $I_+=1$\,. 
This yields 
\begin{align}
\De_g=A_+-2\pi I_+=4\pi-2\pi =2\pi\,. 
\end{align}
Therefore, the total rotation angle is given by 
\begin{align}
\De=\De_d+\De_g=\frac{2\pi a}{b}	+2\pi\,. 
\end{align}

\item[(ii)]
The second case is defined by 
\begin{align}
\theta(t)=2\pi t\,, \qquad \be(t)=\frac{\pi}{2}\,. 
\end{align}
In this case, the dynamical phase reads \dis{\De_d=2\pi a/b\,.}
The curve $\ga$ is the grand circle in $z=0$ plane, 
that is the {\it equator} oriented from west to east.  
The regularized curve  $\ga(\ep)$ is same to $\ga$\,. 
Hence, $A_+$ is the half area of $S^2$\,, namely $2\pi$\,, and $I_+=1$\,. 
See Fig.\,\ref{fig:exam-pihalf}. 
This yields the vanishing of the geometric phase, 
\begin{align}
\De_g=A_+-2\pi I_+=2\pi-2\pi =0\,. 
\end{align}
Then, the total rotation angle is given by 
\begin{align}
\De=\De_d+\De_g=\frac{2\pi a}{b}	+0=\frac{2\pi a}{b}\,. 
\end{align}

\item[(iii)]
The third case is the motion defined by 
\begin{align}
\theta(t)=2\pi t\,, \qquad \be(t)=\pi\,. 
\end{align}
In this case, the dynamical phase is \dis{\De_d=2\pi a/b}\,, too.
The curve $\ga$ is just a point at the {\it north pole} $(0,0,1)$\,, 
and the regularized curve  $\ga(\ep)$ is a small circle enclosing it on the left side. 
Hence, $A_+$ is zero, and $I_+=1$\,. 
See Fig.\,\ref{fig:exam-pi}. 
This gives  
\begin{align}
\De_g=A_+-2\pi I_+=0-2\pi =-2\pi\,. 
\end{align}
Therefore, the total rotation angle is given by 
\begin{align}
\De=\De_d+\De_g=\frac{2\pi a}{b}	-2\pi\,. 
\end{align}
\end{itemize}

\ms 

The above three examples reproduce our observations 
discussed in \secref{subsubsec:obs}\,. 
Next example interpolates these three cases. 

\ms 

\begin{figure}[h]
\centering
\begin{minipage}{0.45\columnwidth}
\centering
\includegraphics[width=1.7\columnwidth]{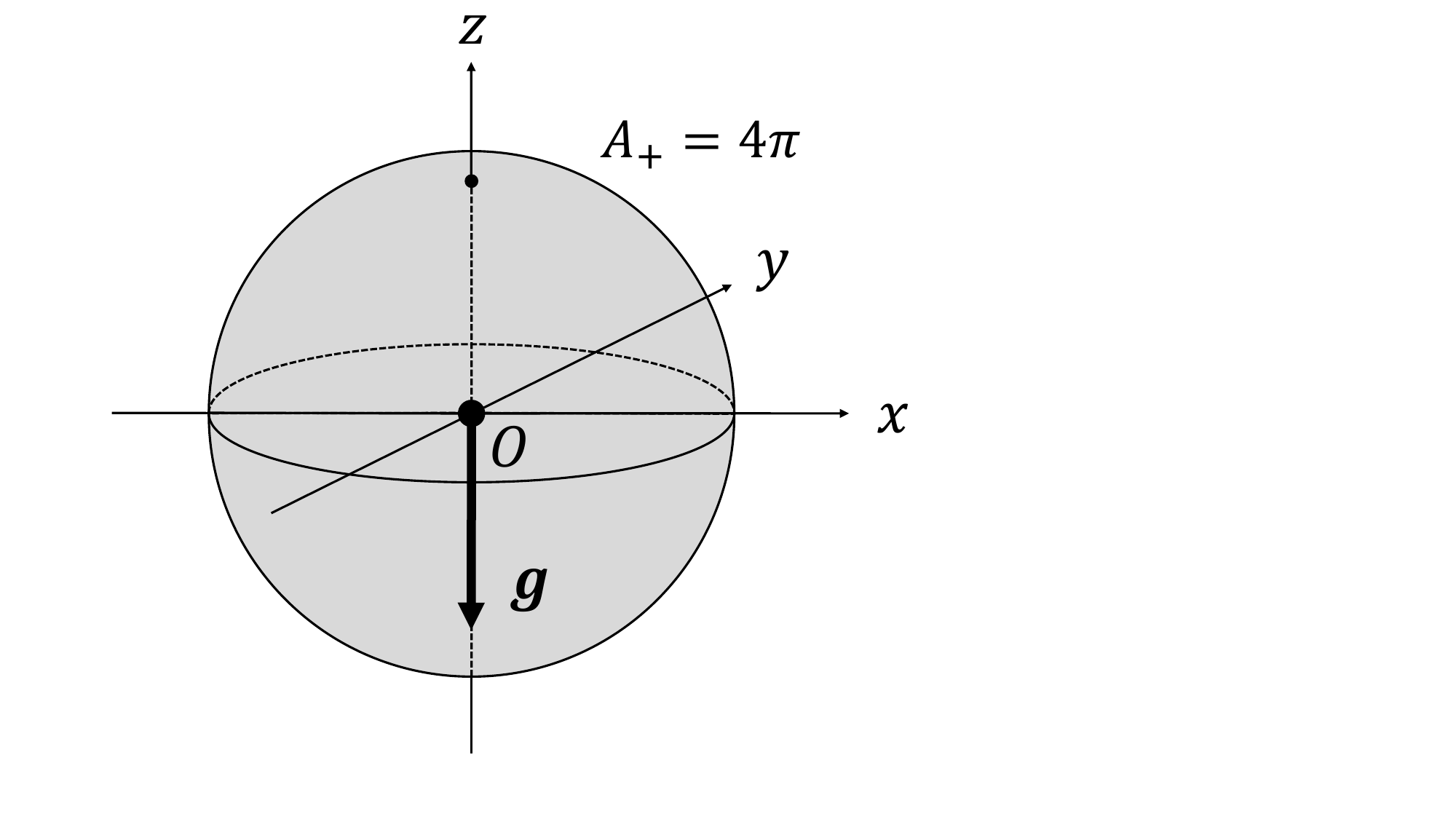}
\caption{Area, for example (i), is $4\pi$\,.
The curve $\ga$ a point at $(0,0,-1)$\,.  
}
\label{fig:exam-zero}
\end{minipage}
\hspace{3mm}
\begin{minipage}{0.45\columnwidth}
\centering
\includegraphics[width=1.7\columnwidth]{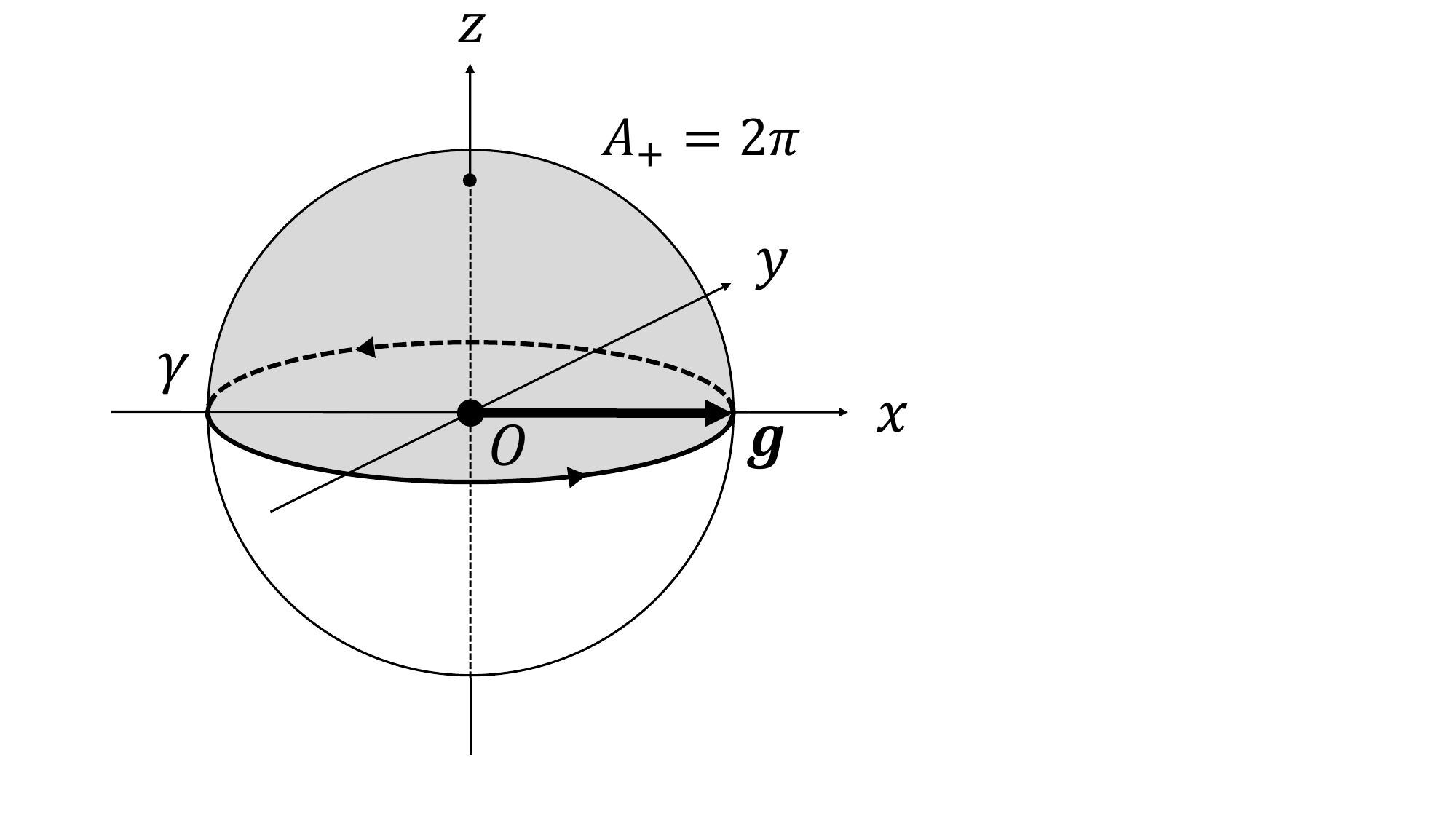}
\caption{Area, for example (ii), is $2\pi$\,.
The curve $\ga$ is the equator.   
}
\label{fig:exam-pihalf}
\end{minipage}
\end{figure}
\begin{figure}[h]
\centering
\begin{minipage}{0.45\columnwidth}
\centering
\includegraphics[width=1.7\columnwidth]{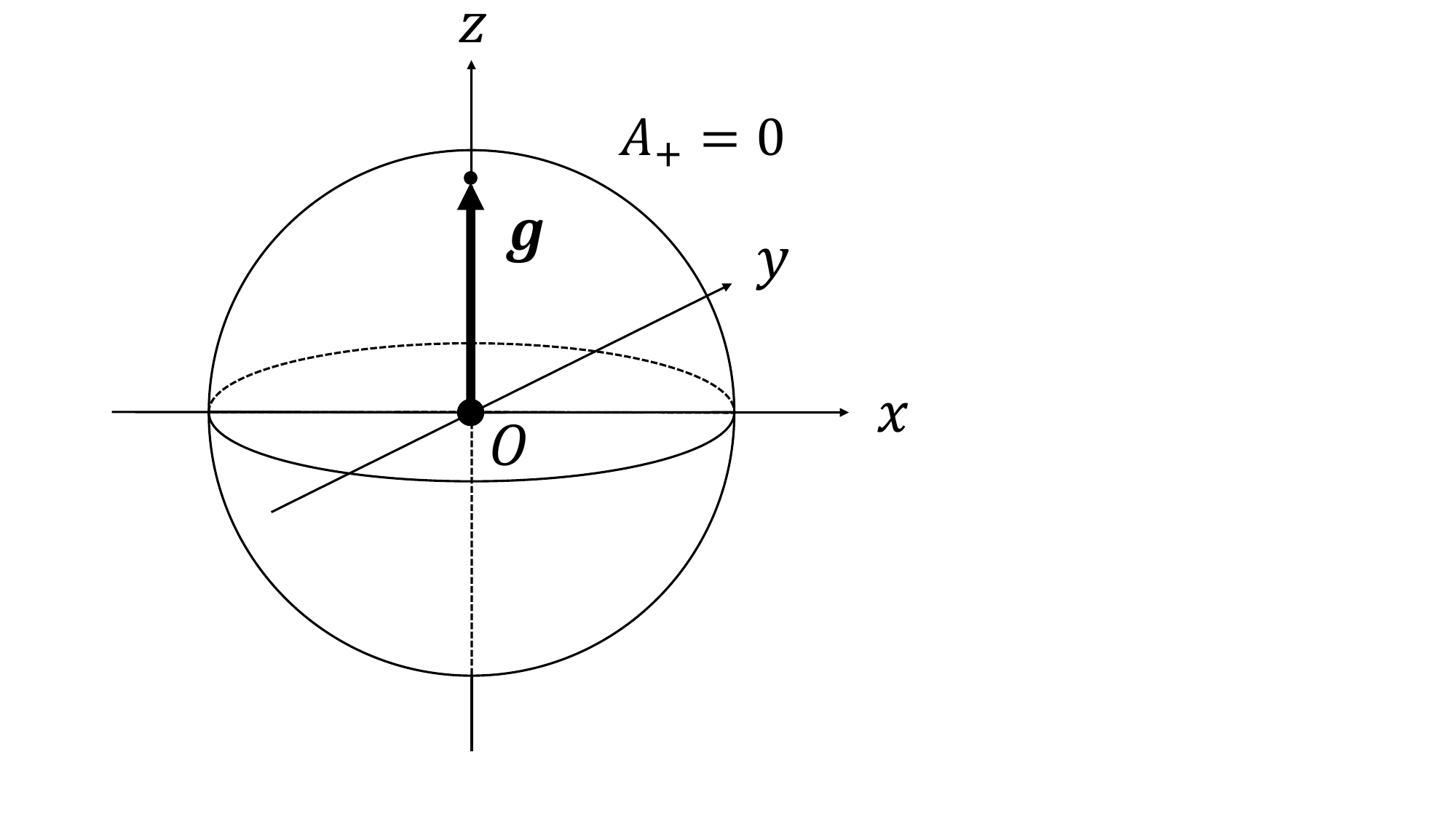}
\caption{Area, for example (iii), is zero.
The curve $\ga$ shrinks to $(0,0,1)$\,.   
\phantom{Area, for example (iii), is zero.}
}
\label{fig:exam-pi}
\end{minipage}
\hspace{3mm}
\begin{minipage}{0.45\columnwidth}
\centering
\includegraphics[width=1.7\columnwidth]{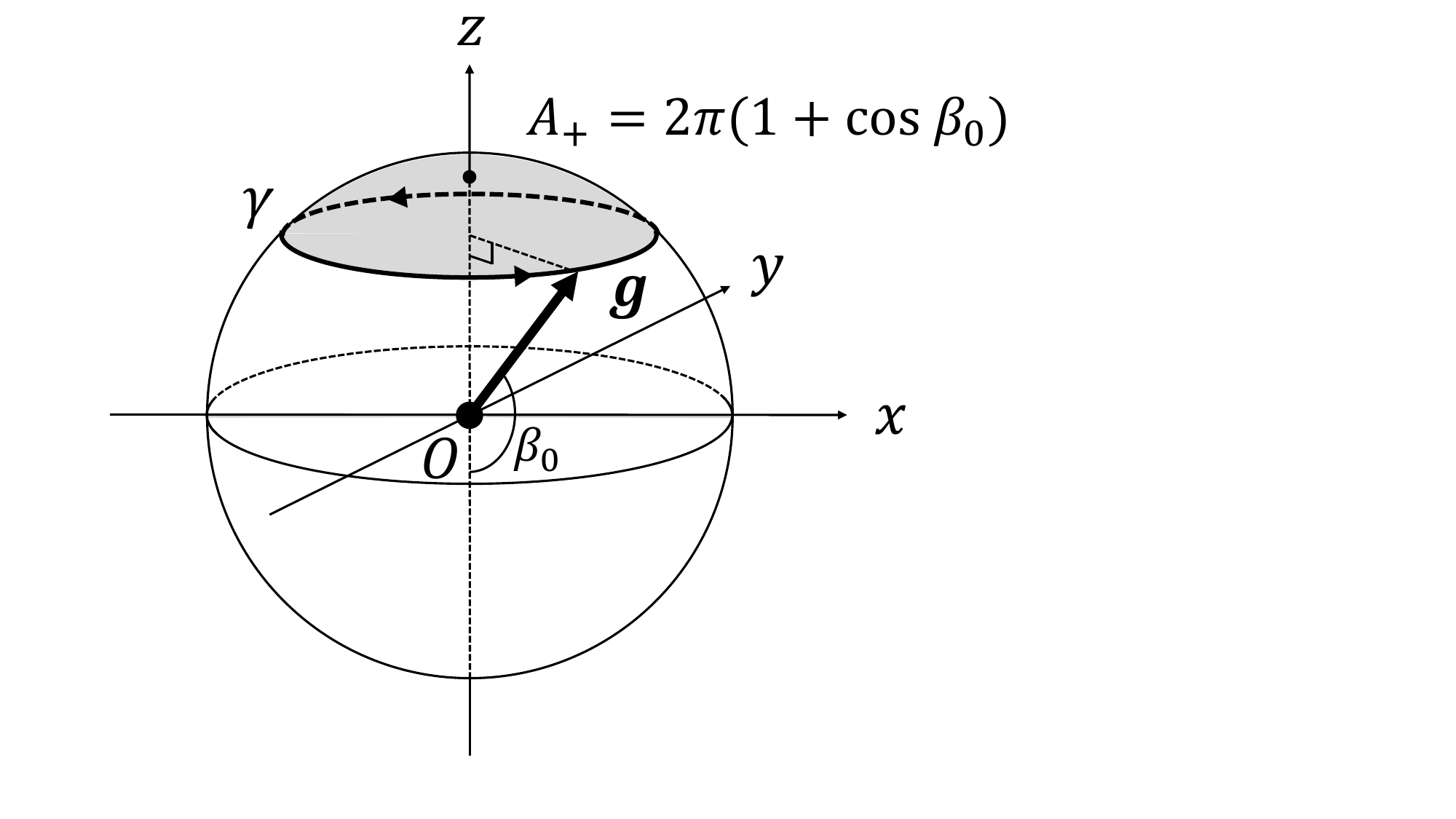}
\caption{Area, for example (iv), is \dis{2\pi(1+\cos\be_0)}\,.
The curve $\ga$ is a circle of latitude.    
}
\label{fig:exam-beta}
\end{minipage}
\end{figure}

\begin{itemize}
\item[(iv)]
The fourth case is the motion defined by 
\begin{align}
\theta(t)=2\pi t\,, \qquad \be(t)=\be_0\,,
\end{align}
where $\be_0\in [0, \pi]$ is constant. 
It is noticed that $\be_0=0, \pi/2, \pi$ correspond to 
the cases (i), (ii), and (iii), respectively. 
In this case, the dynamical phase is also \dis{\De_d=2\pi a/b}\,.
The curve $\ga$ is just a small circle with a constant {\it latitude}. 
The regularized curve  $\ga(\ep)$ is necessary if either $\be_0=0$ or $\pi$. 
See Fig.\,\ref{fig:exam-beta}. 
In this case, $A_+$ is the area of the spherical cap covering the north pole, 
which is calculated as 
\begin{align}
A_+=2\pi(1+\cos\be_0)\,,  
\end{align}
and $I_+=1$\,. 
Hence, we obtain 
\begin{align}
\De_g=A_+-2\pi I_+=2\pi(1+\cos\be_0)-2\pi =2\pi\cos\be_0\,. 
\label{eq:exam-iv-geo}
\end{align}
We will see how this geometric phase is relating to the {\it Foucault's pendulum}
in \secref{subsec:fou}\,. 
Therefore, the total rotation angle is given by 
\begin{align}
\De=\De_d+\De_g=\frac{2\pi a}{b}	+2\pi\cos\be_0\,. 
\end{align}
\end{itemize}

\ms 
The following two examples address cases where $\be$ changes. 
It is noted that they have the same closed curve $\ga$\,,
but the regularized curves $\ga(\ep)$ are different, 
which give rise to the different indices $I_+$\,.  
This is a concrete example of the situation mentioned in 
\remref{rem:ga12}\,. 

\ms 

\begin{itemize}
\item[(v)] 
The fifth case is the motion defined by 
\begin{align}
\label{eq:exam-v}
\theta(t)=\begin{cases}
2\pi t & (0\leq t\leq 1/4) \\
\pi/2 & (1/4\leq t\leq 1/2) \\
2\pi(3/4 -t) &  (1/2\leq t\leq 3/4) \\
0 &  (3/4\leq t\leq 1) 
\end{cases}\,, 
\quad 
\be(t)=\begin{cases}
0 & (0\leq t\leq 1/4) \\
2\pi (t-1/4) & (1/4\leq t\leq 1/2) \\
\pi /2 &  (1/2\leq t\leq 3/4) \\
2\pi(1-t) &  (3/4\leq t\leq 1) 
\end{cases}\,. 
\end{align}
In this case, it is noticed that the dynamical phase is zero, 
\begin{align}
\De_d=\frac{a\theta(1)}{b}=0\,. 
\end{align}
The curve $\ga$ is a curved triangle on $S^2$ 
with one vertex at the south pole $(0,0,-1)$ and the other two on the equator.  
The orientation of $\ga$ on the equator is from east to west.
Every external angle at the corners is $\pi/2$, and 
all edges are parts of the grand circle\,, {\it i.e.}\,, the geodesic lines. 
The regularized curve  $\ga(\ep)$ is a curved triangle 
having a chip at the corner of $(0,0,-1)$\,. 
This gives rise to $I_+=0$ because the south pole {\it is not} 
included in the region $S_+(\ep)$\,.  
See \figref{fig:exam-triang-minus} 
and \figref{fig:exam-triang}\,, 
where we have flipped the direction of $z$-axis for the explanation.  
On the other hand, $A_+$ in 
\figref{fig:exam-triang}
is the $1/8$ area of $S^2$\,, that is,
$4\pi/8=\pi/2$\,. 
Consequently, the geometric phase turns out to be  
\begin{align}
\De_g=A_+-2\pi I_+=\frac{\pi}{2}-0 =\frac{\pi}{2}\,. 
\label{eq:exam-v-gp}
\end{align}
Therefore, the total rotation angle is given by 
\begin{align}
\De=\De_d+\De_g=0+\frac{\pi}{2}=\frac{\pi}{2}\,. 
\end{align}
In this example, it is interesting that $\De$ does not depend 
on the radii of discs $a, b$ due to the vanishing 
of the dynamical phase $\De_d=0$\,,  
and is equal to the geometric phase $\De_g$ itself.

\item[(vi)] 
The sixth case is the motion defined by 
\begin{align}
\label{eq:exam-vi}
\theta(t)=\begin{cases}
-6\pi t & (0\leq t\leq 1/4) \\
-3\pi/2 & (1/4\leq t\leq 1/2) \\
-2\pi(t+1/4) &  (1/2\leq t\leq 3/4) \\
-2\pi &  (3/4\leq t\leq 1) 
\end{cases}\,, 
\quad 
\be(t)=\begin{cases}
0 & (0\leq t\leq 1/4) \\
2\pi (t-1/4) & (1/4\leq t\leq 1/2) \\
\pi /2 &  (1/2\leq t\leq 3/4) \\
2\pi(1-t) &  (3/4\leq t\leq 1) 
\end{cases}\,. 
\end{align}
In comparison with the previous example (v), 
$\be(t)$ is same but $\theta(t)$ is different. 
In this case, the dynamical phase reads 
\begin{align}
\De_d=\frac{a\theta(1)}{b}=-\frac{2\pi a}{b}\,. 
\end{align}
It is noted that 
the curve $\ga$ is exactly the same as the example (v).
However, the regularized curve  $\ga(\ep)$ is a curved triangle 
with a protrusion covering the pole $(0,0,-1)$\,. 
See Fig.\,\ref{fig:exam-triang-plus} and Fig.\,\ref{fig:exam-triang}. 
This yields $I_+=1$ because the south pole {\it is} 
contained in the region $S_+(\ep)$\,.  
Taking into account $A_+=\pi/2$\,, we gain the geometric phase as  
\begin{align}
\De_g=A_+-2\pi I_+=\frac{\pi}{2}-2\pi =-\frac{3\pi}{2}\,. 
\label{eq:exam-vi-gp}
\end{align}
Therefore, the total rotation angle is given by 
\begin{align}
\De=\De_d+\De_g=-\frac{2\pi a}{b}-\frac{3\pi}{2}\,. 
\end{align}
\end{itemize}

\begin{figure}[htbp]
\centering
\begin{minipage}{0.45\columnwidth}
\centering
\includegraphics[width=1.7\columnwidth]{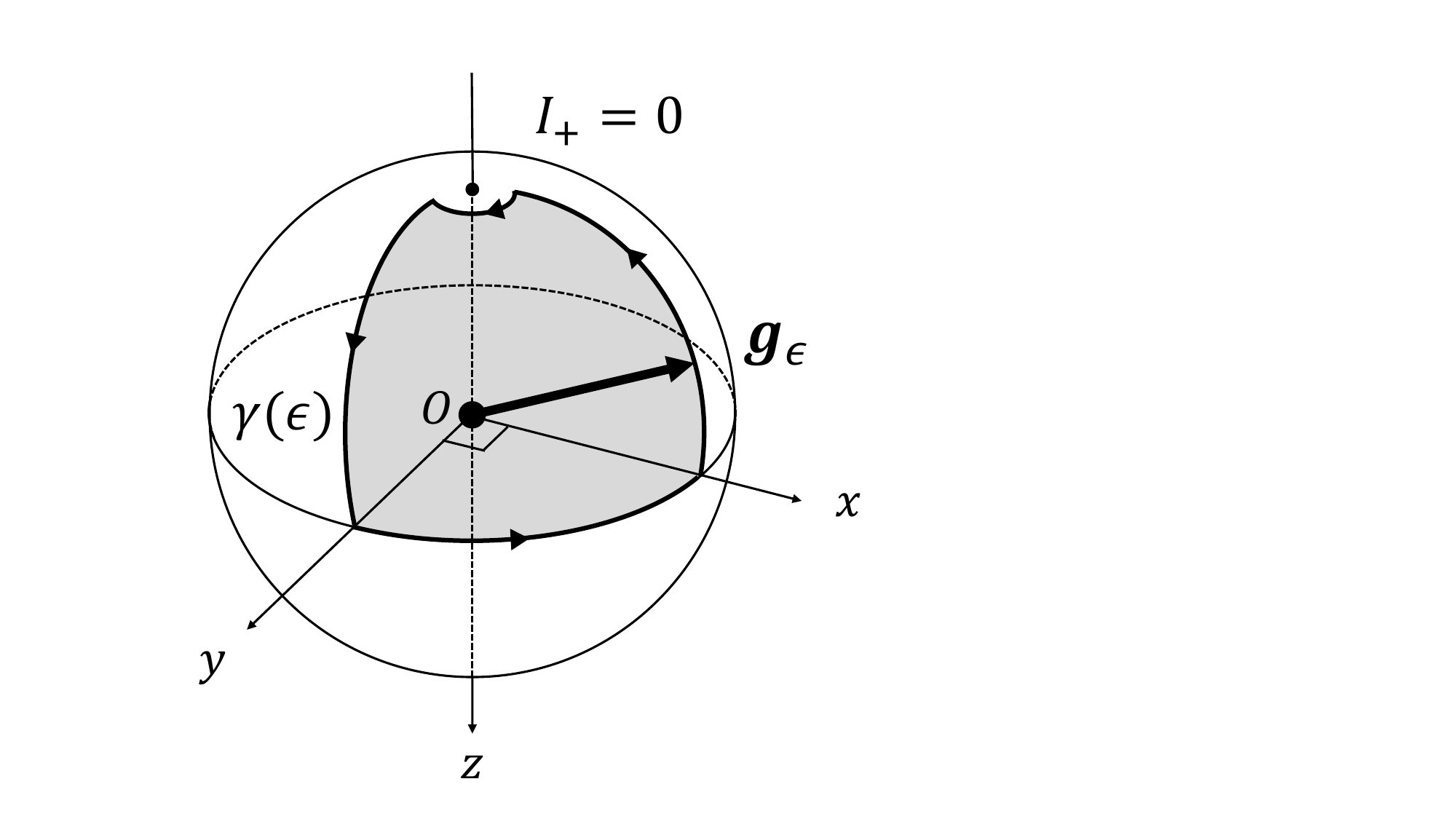}
\caption{The regularized curve $\ga(\ep)$ for example (v). 
The south pole {\it is not} contained in the grayed 
region $S_+(\ep)$\,, 
and hence $I_+=0$.  
}
\label{fig:exam-triang-minus}
\end{minipage}
\hspace{3mm}
\begin{minipage}{0.45\columnwidth}
\centering
\includegraphics[width=1.7\columnwidth]{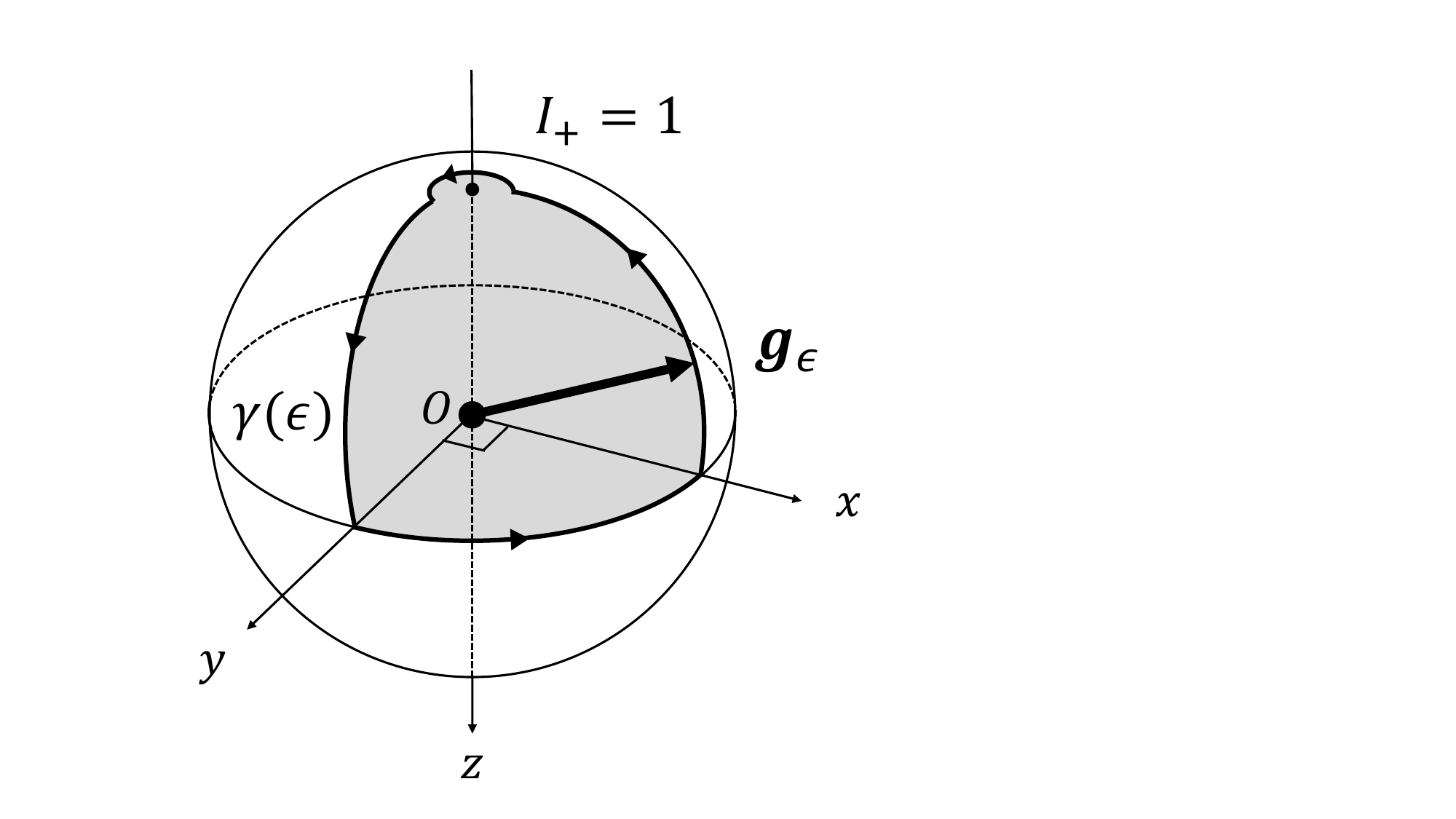}
\caption{The regularized curve $\ga(\ep)$ for example (vi). 
The south pole {\it is} contained in the grayed region 
$S_+(\ep)$\,, and hence $I_+=1$.  
}
\label{fig:exam-triang-plus}
\end{minipage}
\\
\vspace{3mm}
\centering
\begin{minipage}{0.45\columnwidth}
\centering
\includegraphics[width=1.7\columnwidth]{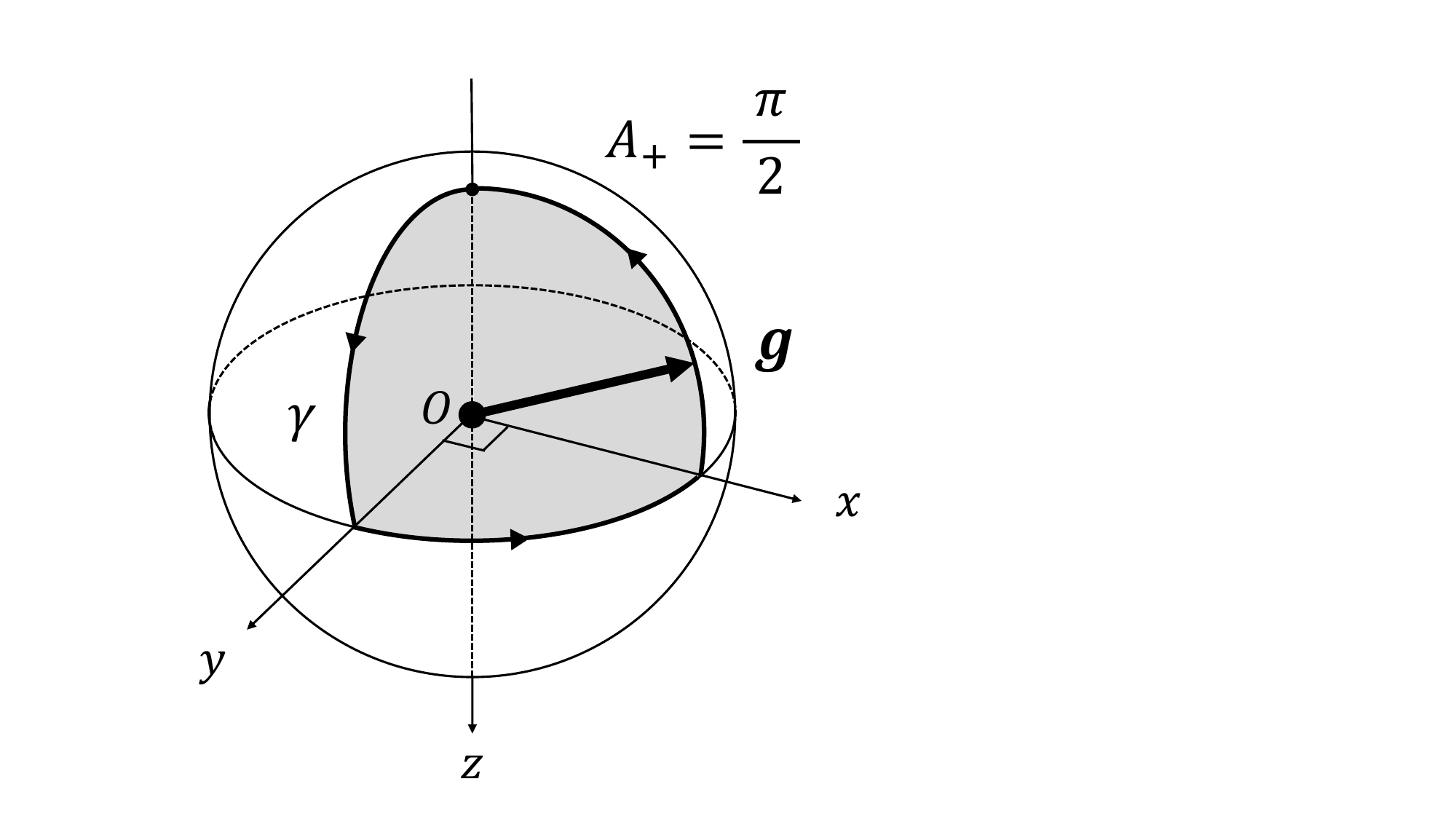}
\caption{Areas, for both examples (v) and (vi), are $\pi/2$\,. 
The curve $\ga$ is a triangle whose edges are the geodesic lines. 
}
\label{fig:exam-triang}
\end{minipage}
\hspace{3mm}
\begin{minipage}{0.45\columnwidth}
\centering
\end{minipage}
\end{figure}

\ms 

\begin{remark}
It is worth directly evaluating  the geometric phases 
for the above examples  (v) and (iv).  
To calculate the integrals given in \dthmref{thm:line-int}\,,
it is sufficient to consider a mesh, 
\begin{align}
0=t_0
<t_1
<t_2
<t_3
<t_4=1\,,
\qquad 
t_1=\frac{1}{4}\,,\quad 
t_2=\frac{1}{2}\,,\quad 
t_3=\frac{3}{4}\,.
\label{eq:vvi-mesh}
\end{align}

\ms 

\begin{itemize}
\item[(v)] 
For the motion defined by \eqref{eq:exam-v}\,, 
it is calculated as 
\begin{align}
\De_g&=\int_{0}^{1} \cos\be(t) \frac{d\theta(t)}{dt} dt 
=\left[\int_{0}^{1/4}+\int_{1/4}^{1/2}
+\int_{1/2}^{3/4}+\int_{3/4}^{1}\right] 
\cos\be(t)  \frac{d\theta(t)}{dt} dt 
\nln
&=
1\times \frac{\pi}{2}
+\int_{1/4}^{1/2}\cos\be(t) \times 0\, dt 
+0\times \left(-\frac{\pi}{2}\right)
+\int_{3/4}^{1} \cos\be(t) \times 0\, dt 
\nln
&= \frac{\pi}{2}\,. 
\end{align}
This indeed coincides with \eqref{eq:exam-v-gp}\,. 
See Fig.\,\ref{fig:exam_v}\,.

\item[(vi)]
For the motion defined by \eqref{eq:exam-vi}\,,
it is similarly computed as 
\begin{align}
\De_g&
=\left[\int_{0}^{1/4}+\int_{1/4}^{1/2}
+\int_{1/2}^{3/4}+\int_{3/4}^{1}\right] 
\cos\be(t)  \frac{d\theta(t)}{dt} dt 
\nln
&=
1\times \left(-\frac{3\pi}{2}\right)
+\int_{1/4}^{1/2}\cos\be(t) \times 0\, dt 
+0\times \left(-\frac{\pi}{2}\right)
+\int_{3/4}^{1} \cos\be(t) \times 0\, dt 
\nln
&= -\frac{3\pi}{2}\,. 
\end{align}
This coincides with \eqref{eq:exam-vi-gp}\,. 
See Fig.\,\ref{fig:exam_vi}\,. 
\end{itemize}

\begin{figure}
\centering
\begin{subfigure}{0.45\columnwidth}
\centering
\includegraphics[width=1.3\columnwidth]{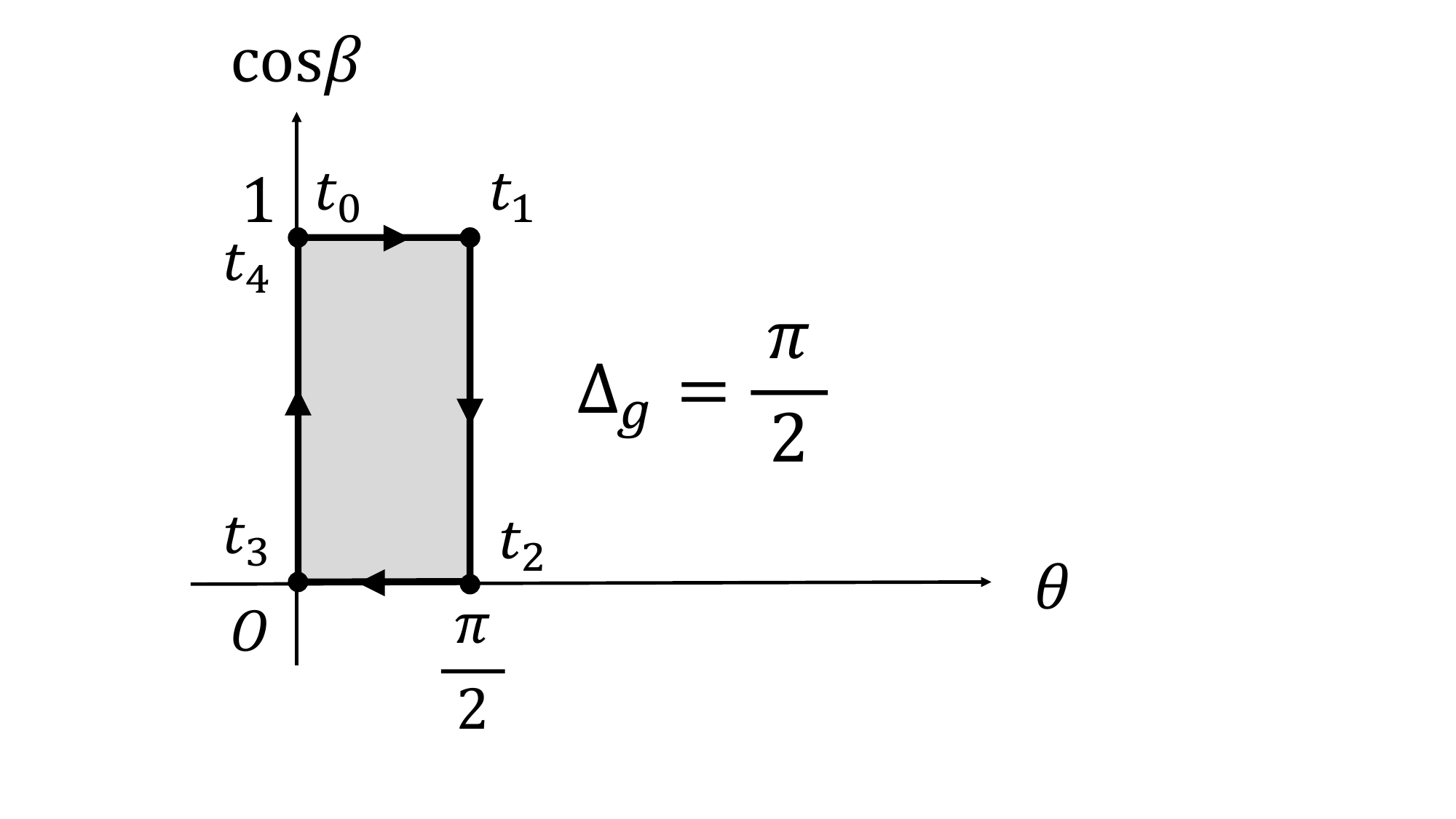}
\caption{The line integral for example (v)}
\label{fig:exam_v}
\end{subfigure}
\begin{subfigure}{0.45\columnwidth}
\centering
\includegraphics[width=1.3\columnwidth]{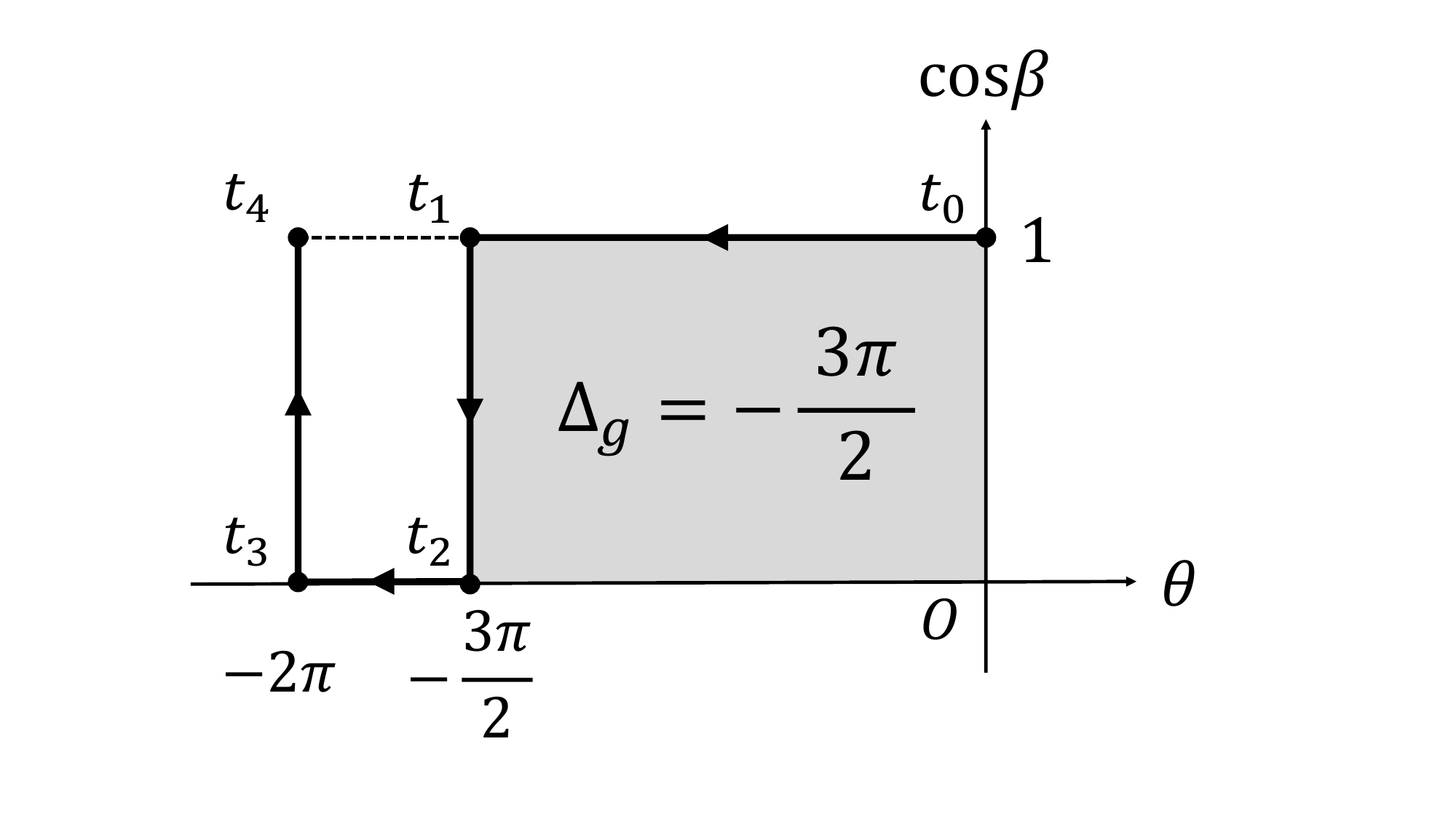}
\caption{The line integral for example (vi)}
\label{fig:exam_vi}
\end{subfigure}
\caption{The signed areas indicated in gray 
give the geometric phases $\De_g$\,. 
The time evolution is from $t_0=1$ to $t_4=1$
as  in \eqref{eq:vvi-mesh}\,. }
\end{figure}
\end{remark}

\section{Some models sharing the common structure}
\label{sec:com}
\setcounter{equation}{0}

In this section, we shall elucidate that the several models 
exhibiting the geometric phases share the common structure
with our model, which include
the {\it Foucault's pendulum} (\ssecref{subsec:fou}), 
{\it Dirac's monopole potentials} (\ssecref{subsec:mono}),
and {\it  Berry phases} (\ssecref{subsec:berry}).   
In other words, our model extracts the underlying mathematical 
structures for these models.

\subsection{Foucault's pendulum and its generalization}
\label{subsec:fou}

Consider a sufficiently long and heavy pendulum suspended 
from the high ceiling so that the oscillation would not be 
damped by the effect of friction. 
If we set the pendulum at the north pole, 
the plane of oscillation spontaneously 
rotates clockwise by $2\pi$ in $24$ hours, 
because the earth itself rotates counterclockwise by $2\pi$ in a day.  
At the south pole, the pendulum rotates counterclockwise
by $2\pi$ in $24$ hours, 
while, on the equator, it remains in the same plane.  

\ms 

This is a famous experiment to demonstrate the earth's rotation,
and often referred to as the {\it Foucault's pendulum} \cite{Fou}. 
In general, the oscillation plane of the Foucault's pendulum 
at the latitude 
$\la\in [-\pi/2, \pi/2]$ rotates clockwise by 
\begin{align}
\De_{\text{Fou}}=2\pi \sin \la \,,  
\label{eq:fou-sin}
\end{align}
in $24$ hours, which is known as the {\it Foucault's sine law}. 

\ms

Our model can be mapped to Foucault's pendulum
through the following steps. 
\begin{enumerate}
\item By parallel transportation in $\RR^3$\,, 
we move disc B so that its center is at 
$\bg(\theta, \be)\in S^2$ and it is tangent to $S^2$
at the point. 

\item Assume $S^2$ as the earth and 
ignore the dynamical phase $\De_d$\,. 
We can regard disc B as the local frame of an observer 
at $\bg(\theta, \be)\in S^2$\,,  
namely a {\it compass} beneath the pendulum. 

\item On the disc B, set a Foucault's pendulum. 
Then, the plane of the pendulum oscillation defines 
the local inertial frame on it. 
\end{enumerate}


By the above procedure, 
our geometric phase $\De_g$ can be identified with 
Foucault's rotation angle $\De_{\rm Fou}$\,. 
Precisely, $\De_g$ expresses how disc B has been forced to 
rotate counterclockwise observed from $-\bg$ (\figref{fig:geo}). 
On the other hand, 
$\De_{\rm Fou}$ means how the oscillation plane,
or equivalently the local inertial frame, on disc B 
rotates clockwise measured from $+\bg$ (\figref{fig:fou}). 
Thus, $\De_g$ and $\De_{\rm Fou}$ measure the same angle 
in the opposite directions as illustrated in \figref{fig:geofou}\,.  
The explicit relations of other quantities between our model 
and Foucault's pendulum
are listed in \tabref{tab:fou}\,.

\begin{figure}
\centering
\begin{subfigure}{0.45\columnwidth}
\centering
\includegraphics[width=\columnwidth]{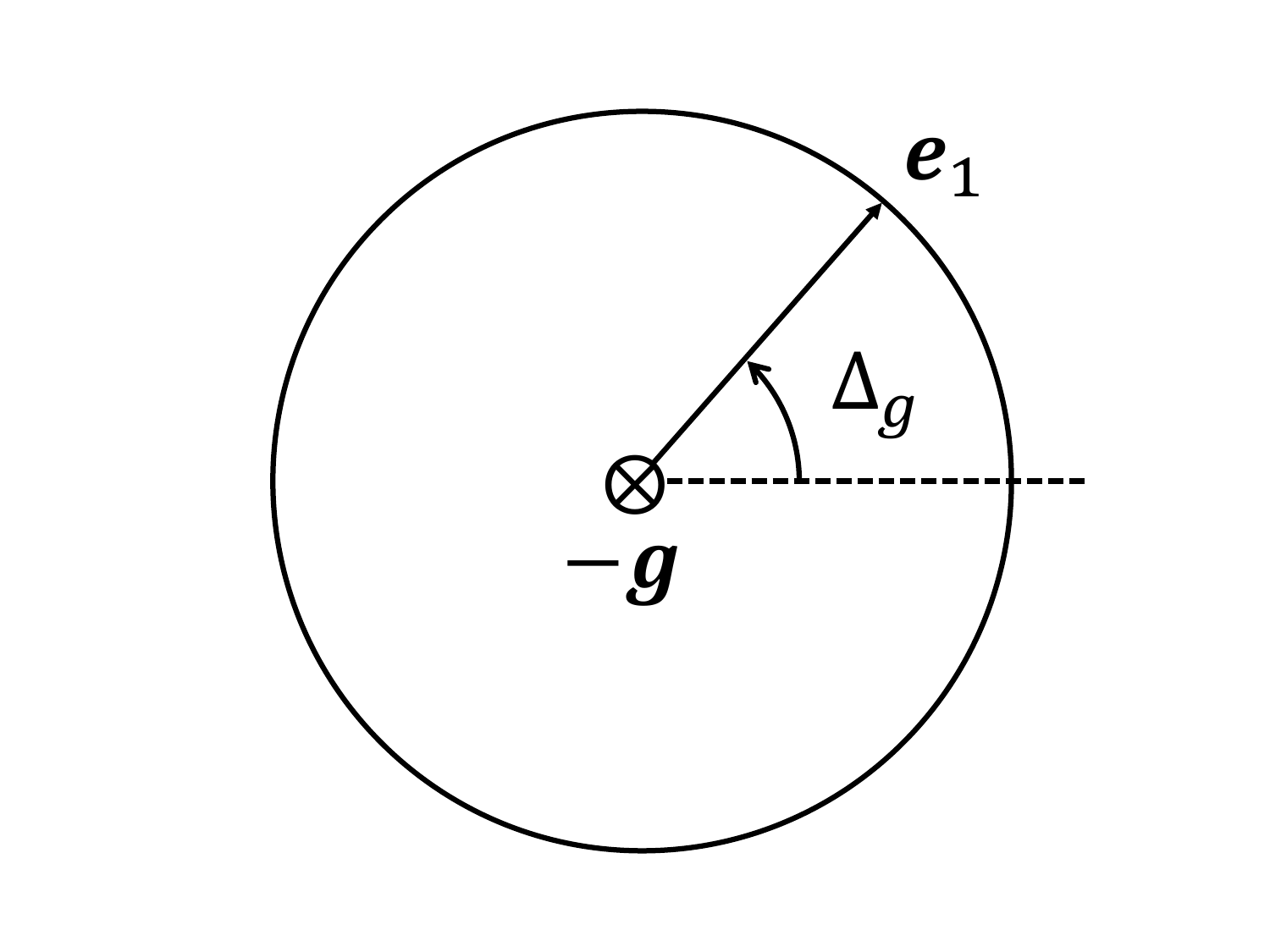}
\caption{The geometric phase, $\De_g$\,. }
\label{fig:geo}
\end{subfigure}
\begin{subfigure}{0.45\columnwidth}
\centering
\includegraphics[width=\columnwidth]{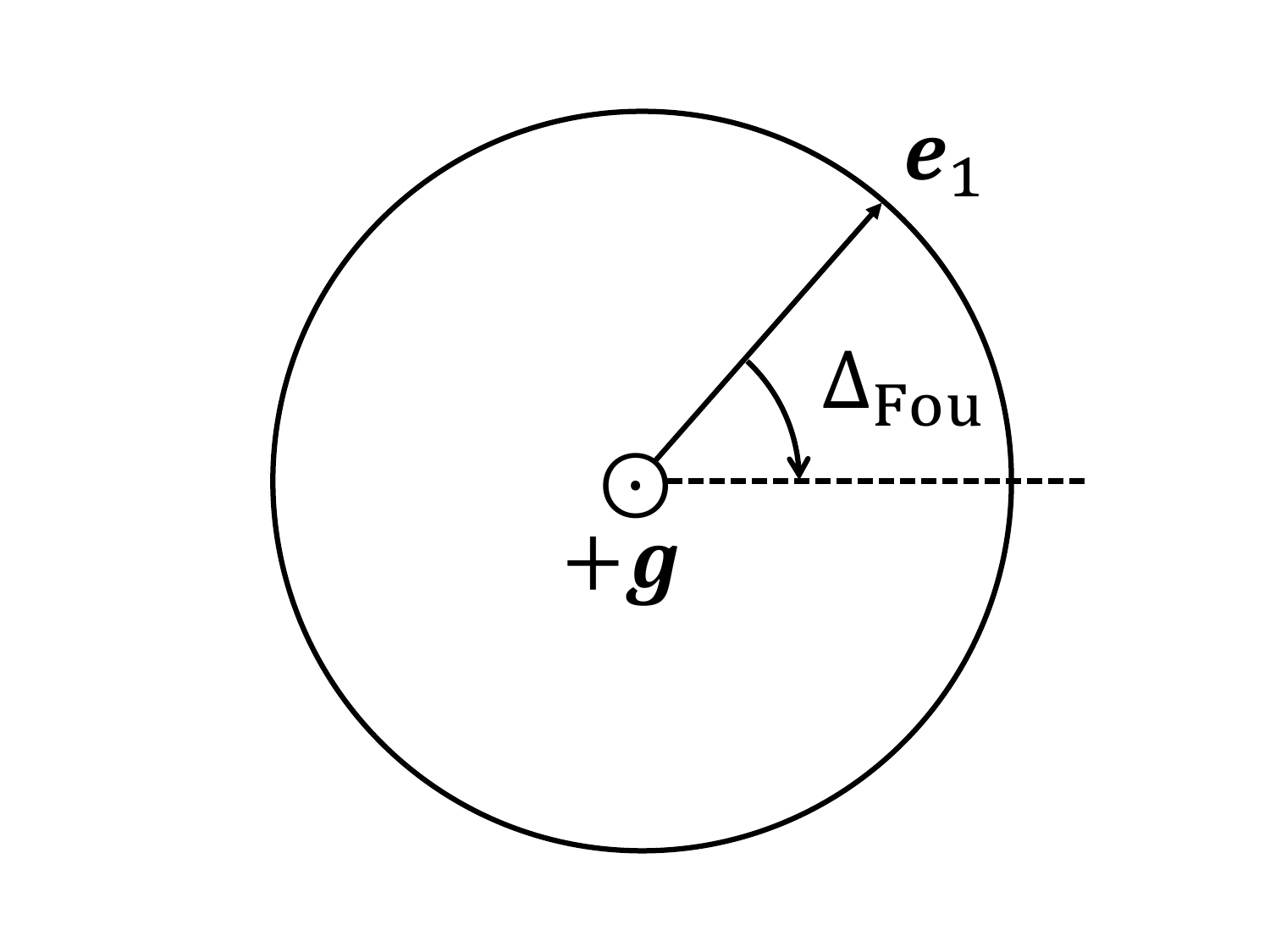}
\caption{The shift of the oscillation plane, $\De_{\rm Fou}$\,. }
\label{fig:fou}
\end{subfigure}
\caption{The relation between $\De_g$ and $\De_{\rm Fou}$
on disc B. The dotted line means the local inertial system, 
in which Foucault's pendulum remains.}
\label{fig:geofou}
\end{figure}
\begin{prop}
\label{prop:fou}
Let $\De_{\rm Fou}$ be the rotation angle of the Foucault's
pendulum. Then, it is related to the geometric phase as 
\begin{align}
\De_{\rm Fou}=-\De_g\, . 
\end{align}
\end{prop}

\ms 

To restore the Foucault's sine law, 
we map the local coordinates $\theta$ and $\be$ to 
the longitude $\la\in [0, 2\pi]$ and the latitude 
$\phi\in [-\pi/2, \pi/2]$\,, respectively, 
as in \tabref{tab:fou},  
\begin{align}
\theta=\la+2\pi t\,,\qquad \be=\phi+\frac{\pi}{2}\,,   
\label{eq:lola}
\end{align}
where time parameter $t$ counts days. 
The factor $2\pi t$ means that the earth rotates by $2\pi$
in a day. 
Note that the motion realizing Foucault's pendulum corresponds to 
the example (iv) in \secref{subsec:ex} and \figref{fig:exam-beta}. 
Since $t$ changes from $0$ to $1$ in a day, 
and $\la, \phi$ are constant, we have 
\begin{align}
\De_{\rm Fou}=-\De_g
=-\int^{1}_0 \cos\left(\phi+\frac{\pi}{2}\right) 
\frac{d(\la+2\pi t)}{dt} dt
=2\pi \sin \phi\,.  
\end{align}
This coincides with Foucault's sine law \eqref{eq:fou-sin}\,. 
\begin{table}[htbp]
\centering
\setlength{\tabcolsep}{5pt}
\caption{The correspondence between our model 
and the Foucault's pendulum}
\begin{tabular}{|c|c|c|c|c|c|} \hline
Our model & Foucault's pendulum \\ \hline\hline 
$S^2$ & the earth \\ \hline
$\De_d$ & $0$ \\ \hline
disc B & 
\begin{tabular}{c} 
local frame of observer  \\ 
(a compass below the pendulum)
\end{tabular}
\\ \hline
$t$ & $t$ (day)\\ \hline
$\theta-2\pi t$ & $\la$ (longitude)  \\ \hline
$\beta-\pi/2$ & $\phi$ (latitude)   \\ \hline
$-\De_g$ & $\De_{\text{Fou}}$ \\ \hline
\end{tabular}
\label{tab:fou}
\end{table}%

\subsection*{Generalized Foucault's pendulum}
\propref{prop:fou} implies the generalization of Foucault's pendulum. 
For instance, assume that a pendulum is carried on a big ship, 
which is still sufficient in the Pacific Ocean.  
The ship is sailing in the north and south directions and back to the 
original position $24$ hours later.
In this case, the sea route looks like a wavy line $\ga$ 
as \figref{fig:mainthm}. 
Taking into account \thmref{thm:main}, 
we gain the rotation angle of the pendulum
against a compass on the ship as 
\begin{align}
\De_{\rm Fou}=-\De_g=2\pi-A_+=A_--2\pi 
=\frac{A_--A_+}{2}\,, 
\end{align}
where $A_+$ is the area of the north side enclosed by 
the sea route $\ga$, 
while $A_-$ is that of the south side. 
This naturally explains why the oscillation plane of the pendulum 
on the equator does not change.  
When the sea route $\ga$ coincides with the equator, 
$A_+$ and $A_-$ are the areas of the 
northern and southern hemispheres, respectively.  
Thus, due to $A_+=A_-=2\pi$\,, we get 
\begin{align}
\De_{\rm Fou}=\frac{A_--A_+}{2}=\frac{2\pi-2\pi}{2}=0\,. 
\end{align}
It is emphasized that this is valid even if 
$\ga$ does {\it not} coincide with the equator. 

\ms 

We believe that our model has potential applications for 
the gyroscopes of airplanes or ships
because \propref{prop:fou} serves the theoretical prediction
of the rotation angles of gyroscopes for a given route $\ga$
referred to a compass (see, for instance, \cite{GR,B}). 
Explicitly, plugging the relations \eqref{eq:lola} with
\propref{prop:fou}, we propose the following 
{\it generalized Foucault's sine law}. 
\begin{prop}
Set Foucault's pendulum on a ship floating 
at the longitude $\la(t)\in [0,2\pi]$ and 
the latitude $\phi(t)\in [-\pi/2, \pi/2]$\,, 
where $t$ is a time parameter counting days. 
Suppose that the ship navigates 
quasistatically from $t=0$ to $t=T$\,. 
During $T$ days, the plane of the pendulum oscillation 
rotates clockwise referred to a compass by 
\begin{align}
\De_{\rm Fou}&=\int^{T}_{0} \sin\phi(t) 
\left(\frac{d\la(t)}{dt}+2\pi\right) dt\,. 
\end{align}
\end{prop}

\ms 

\begin{remark}
H.~Kamerlingh Onnes also investigated Foucault's pendulums, 
and the result is collected in 
his Ph.D. dissertation \cite{On}, see also \cite{Sch}\,. 
The geometric aspects of Foucault's pendulum are argued   
in \cite{MOWZ, MPM}\,, where the relation to Hopf fibration 
is also mentioned.  
The geometric approach for mechanics 
is addressed in \cite{Mar, MR, BB,Reeb}\,. 
\end{remark}

\subsection{Dirac's monopole potential}
\label{subsec:mono}

The geometric phase $\De_g$ in our model can also be realized by the holonomy of 
the $U(1)$ principal bundle, which is so-called the {\it Dirac's monopole potential}. 

\ms 

We use the following local orthonormal frame given in \eqref{eq:e123},  
\begin{align}
\ee_1
=\begin{pmatrix} -\sin \theta \\ \cos \theta \\ 0  \end{pmatrix}\,, 
\quad 
\ee_2
=\begin{pmatrix} 
\cos\be \cos \theta \\ \cos\be \sin \theta \\ \sin\be  \end{pmatrix}\,, 
\quad 
\ee_3
=\begin{pmatrix}\sin\be \cos\theta \\ \sin\be \sin\theta \\ -\cos\be \end{pmatrix}\,.  
\end{align}
Note that $\ee_3=\bg(\theta, \be)$\,. 
The relations between the polar coordinates \dis{(r, \be, \theta)}
and the cartesian coordinates \dis{(x,y,z)} read 
\begin{align}
x&= r \sin\be \cos\theta \,, \notag \\
y&= r \sin\be \sin\theta \,, \notag \\
z&= -r \cos\be \,. 
\label{eq:polar}
\end{align}

\ms 

\begin{definition}[Dirac's monopole potentials \cite{D, EGH}]
The {\it Dirac's monopole potentials} are vector fields defined by 
\begin{align}
\bA_\pm=\frac{\pm1+\cos\be}{r\sin\be}\ee_1 
=\frac{1}{r(r\pm z)}\begin{pmatrix} -y \\ x \\ 0  \end{pmatrix}\,, 
\label{eq:mono}
\end{align}
where $\bA_+$ is well-defined for $\be\neq 0$ and $z\neq -r$\,, while 
$\bA_-$ is for $\be\neq \pi$ and $z\neq r$\,. 
\end{definition}

\ms 

Here, we have set the magnetic charge to one. 
By definition, $\bA_+$ and $\bA_-$ have stringy singularities 
on $z\leq 0$ and $z\geq 0$ of $z$-axis, respectively. 
The following proposition tells that those potentials describe the magnetic fields
generated by the magnetic monopole at the origin. 
\begin{prop}
\label{prop:vp}
For the vector fields $\bA_\pm$\,, it holds that 
\begin{align}
\nabla\times \bA_+=\frac{1}{r^2}\ee_3 \quad (\be\neq 0)\,, 
\notag \\
\nabla\times \bA_-=\frac{1}{r^2}\ee_3 \quad  (\be\neq \pi)\,. 
\end{align}
\end{prop}

\ms 

Next, let us consider how our model is relating to the monopole 
potentials. 
Recalling the connection matrix \eqref{eq:cmat}, we can write  
\begin{align}
d\theta
=\frac{\om_{31}}{\sin \be} 
= \frac{1}{\sin \be}\, \ee_1\cdot d\bg\,.  
\end{align}
Applying this expression for Prop.\,\ref{prop:gp-reg}, 
we get 
\begin{align}
\De_g
=\lim_{\ep\to0}\int_{\ga(\ep)} \cos\be_\ep\, d\theta
=\lim_{\ep\to0}\int_{\ga(\ep)} \frac{\cos\be_\ep}{\sin\be_\ep}\, 
\ee_1 \cdot d\bg_\ep\,. 
\end{align}
It is noted that $\sin\be_\ep\neq0$ for $\ga(\ep)$\,. 
By definition \eqref{eq:mono}, it holds that  
\begin{align}
\frac{\bA_++\bA_-}{2}
&=\frac{\cos\be}{r\sin\be}\ee_1\,, 
& 
\frac{\bA_+-\bA_-}{2}
&=\frac{1}{r\sin\be}\ee_1 
=\nabla \theta\,.  
\end{align}
Here, the second relation means the the {\it gauge transformation}
between $\bA_+$ and $\bA_-$ on the strip covering the equator,  
which follows from 
the gradient operator in the polar coordinates 
\begin{align}
\nabla=\ee_1\frac{1}{r\sin\be}\frac{\partial }{\partial \theta} 
+\ee_2\frac{1}{r}\frac{\partial }{\partial \be} 
+\ee_3\frac{\partial }{\partial r} \,. 
\end{align}
Using these relations with $r=1$, we obtain 
\begin{align}
\De_g
&=\lim_{\ep\to0}\int_{\ga(\ep)} 
\frac{\bA_++\bA_-}{2} 
\cdot d\bg_\ep
=\lim_{\ep\to0}\int_{\ga(\ep)} 
(\bA_\pm\mp\nabla \theta) \cdot d\bg_\ep\,. 
\end{align}

\ms 

We see that the term stemming from the gauge transformation 
yields the topological number \eqref{eq:top} because of 
\begin{align}
\lim_{\ep\to0}\int_{\ga(\ep)} 
\nabla \theta \cdot d\bg_\ep
=
\lim_{\ep\to0}\int_{\ga(\ep)} d\theta
=\theta(1)-\theta(0)
=2\pi n\,. 
\end{align}

\ms 

In summary, we arrive at the following proposition. 
\begin{prop}
The geometric phase is obtained by the monopole potentials 
defined in \eqref{eq:mono} as 
\begin{align}
\De_g
&=\lim_{\ep\to0}\int_{\ga(\ep)} 
\frac{\bA_++\bA_-}{2} 
\cdot d\bg_\ep
\nln
&=\lim_{\ep\to0}\int_{\ga(\ep)} 
\bA_+\cdot d\bg_\ep-2\pi n  
\nln
&=\lim_{\ep\to0}\int_{\ga(\ep)} 
\bA_- \cdot d\bg_\ep+2\pi n\,,   
\label{eq:Dg-mono}
\end{align}
where $n$ is the topological number given in \eqref{eq:top}\,,
$\bg_\ep$ is the deformed Gauss vector in \eqref{eq:gauss-ep}\,,
and $\ga(\ep)\in S^2-\{(0,0,\pm1)\}$ is the regularized curve
in \eqref{eq:ga-reg}\,. 
\end{prop}

\ms 

Concluding this subsection, let us confirm the consistency 
of the above proposition with Thm.\,\ref{thm:main}. 
Suppose that $\ga$ is a simple closed curve enclosing
$(0,0,1)$ on the left side as Fig.\,\ref{fig:mainthm}. 
In this case, noting $\ga(\ep)=\ga$ and 
$\partial S_+=-\partial S_-=\ga$\,, we have  
\begin{align}
\De_g
=\lim_{\ep\to0}\int_{\ga(\ep)} 
\frac{\bA_++\bA_-}{2} 
\cdot d\bg_\ep
=\frac{1}{2}\left[\int_{\partial S_+} \bA_+ \cdot d\bg
-\int_{\partial S_-}\bA_- \cdot d\bg \right]\,. 
\end{align}
Since  these terms are equal to the areas $S_\pm$ by 
Stokes' theorem and Prop.\,\ref{prop:vp} with $r=1$\,, 
\begin{align}
\int_{\partial S_\pm} \bA_\pm \cdot d\bg
=\int_{S_\pm} (\nabla\times \bA_\pm) \cdot \ee_3\, dS_\pm
=A_\pm\,, 
\end{align}
we obtain the expected result 
\begin{align}
\De_g=\frac{A_+-A_-}{2}\,.   
\end{align}
This coincides to \eqref{eq:mainthm}
with $(I_+,I_-)=(1,1)$\,. 

\ms 

As another example, 
consider the case $(I_+,I_-)=(0,2)$\,, 
in which the curve $\ga$ does not contain any poles
$(0,0,\pm1)$ on the left side. 
In this case, the topological number vanishes, $n=0$\,. 
Hence, from the second relation in \eqref{eq:Dg-mono}\,,
we get 
\begin{align}
\De_g=\int_{\ga=\partial S_+} \bA_+\cdot d\bg_\ep-0 
=A_+=4\pi-A_-\,. 
\end{align}
Again, this is consistent with \eqref{eq:mainthm}\,. 

\ms 

\begin{remark}[A simple derivation of the Dirac' monopole potential]
Though the Dirac's monopole potentials \eqref{eq:mono}
are well-known in theoretical physics and the theory of 
fiber bundles, it is worth addressing  
the heuristic derivation in our context.  
Let us seek a solution for the equation, 
\begin{align}
\nabla\times \bA=\frac{1}{r^2}\ee_3 
\end{align}
with the following ansatz 
\begin{align}
\bA=f(\theta,\be) \ee_1\,,  
\label{eq:ansatz}
\end{align}
where $f(\theta,\be)$ is a scalar function to be determined. 
Applying Stokes's theorem for the spherical cap $S_+$ covering the north
pole of a sphere with radius $r$, 
we see that  
\begin{align}
\int_{\partial S_+} \bA \cdot \ee_1 ds 
=\int_{S_+} (\nabla\times \bA) \cdot \ee_3\, dS_+\,. 
\end{align}
For $r=1$\,, see Fig.\,\ref{fig:exam-beta}. 
By the ansatz \eqref{eq:ansatz}\,, both hand sides 
are simplified as follows, 
\begin{align}
\text{(LHS)}
&
=f(\theta,\be) \int_{\partial S_+} ds 
=f(\theta,\be)\times \text{(length of the edge of the cap)}\,, 
\nln 
\text{(RHS)}
&
=\frac{1}{r^2} \int_{S_+} dS_+
=\frac{1}{r^2}\times \text{(area of the cap)}\,.  
\end{align}
Equating these, we find the scalar factor as
\begin{align}
f(\theta,\be) 
=\frac{\text{(area of the cap)} }
{r^2\times  \text{(length of the edge of the cap)}}
=\frac{2\pi r^2(1+\cos\be)}{r^2\times 2\pi r\sin\be}\,. 
\end{align}
This restores 
\begin{align}
\bA=\frac{1+\cos\be}{r \sin\be}\ee_1=\bA_+\,. 
\end{align}
The parallel argument for the spherical cap covering the south pole
will give us $\bA_-$\,. 
\end{remark}

\subsection{Berry phases}
\label{subsec:berry}

Let us clarify the relation of our model with 
the {\it Berry phase} \cite{Pan,Lon,han,Berry1,Berry2}, 
which is a typical geometric phase 
in quantum mechanics. 

\ms 

Consider the two-level system with the Hamiltonian, 
\begin{align}
H=\frac{1}{r}\begin{bmatrix} z & x-iy \\ x+iy & -z 
\end{bmatrix}
=\begin{bmatrix} -\cos\be & e^{-i\theta}\sin\be \\ 
e^{i\theta}\sin\be  &  \cos\be
\end{bmatrix}\,, 
\end{align}
where we have used the polar coordinates defined by 
\eqref{eq:polar}\,. 
The eigenvalues of $H$ are $\pm1$\,. 
The eigenvectors with eigenvalue $1$ are given by 
\begin{align}
\ket{\psi_+}&=\frac{1}{\sqrt{2r(r+z)}} 
\begin{bmatrix} r+z \\ x+iy \end{bmatrix}
=\begin{bmatrix}
\sin\frac{\be}{2} \\  e^{i\theta}\cos\frac{\be}{2} 
\end{bmatrix}\qquad 
(z\neq -r\,, \be\neq 0)\,, 
\notag \\
\ket{\psi_-}&=\frac{1}{\sqrt{2r(r-z)}} 
\begin{bmatrix} x-iy \\ r-z \end{bmatrix}
=\begin{bmatrix}
e^{-i\theta}\sin\frac{\be}{2} \\ 
\cos\frac{\be}{2} 
\end{bmatrix}\qquad 
(z\neq r\,, \be\neq \pi)\,. 
\end{align}
These two vectors are only different by the phase factor, 
\begin{align}
\ket{\psi_+}=e^{i\theta}\ket{\psi_-}\,. 
\end{align}
Since the vector $\ket{\psi_+}$ is singular at $z=-r$ or 
$\be=0$\,, 
it is well-defined on $S^2-\{(0,0,-1)\}$\,. 
While, the vector $\ket{\psi_-}$ is singular at $z=r$ or 
$\be=\pi$\,, 
and defined on $S^2-\{(0,0,1)\}$\,. 

\ms 

The {\it Berry connection} is then calculated as 
\begin{align}
\label{eq:berry}
\vev{\psi_+|d|\psi_+}&=\frac{i}{2}(+1+\cos\be)d\theta\qquad 
(\be\neq 0)\,, 
\notag \\
\vev{\psi_-|d|\psi_-}&=\frac{i}{2}(-1+\cos\be)d\theta\qquad 
(\be\neq \pi)\,,  
\end{align}
where the hermitian inner product is defined by 
\begin{align}
\vev{a|b}=\overline{a}_1b_1+\overline{a}_2b_2 
\quad\text{for}\quad 
\ket{a}=\begin{bmatrix} a_1 \\ a_2  \end{bmatrix}, ~
\ket{b}=\begin{bmatrix} b_1 \\ b_2  \end{bmatrix}
\in \CC^2\,.  
\end{align}
These coincide with the Dirac's monopole potentials given in \eqref{eq:mono}\,, 
\begin{align}
\vev{\psi_\pm|d|\psi_\pm}=\frac{i}{2}\bA_\pm\cdot d\bg\,. 
\end{align}

\ms 

Hence, for $\be\neq 0, \pi$\,, it allows us to express 
one-form $\cos\be d\theta$ in terms of the Berry connections, 
\begin{align}
\cos\be d\theta
&= -i \vev{\psi_+|d|\psi_+}-i\vev{\psi_-|d|\psi_-}
\nln
&=-2i \vev{\psi_+|d|\psi_+}-d\theta 
\nln
&=-2i \vev{\psi_-|d|\psi_-}+d\theta\,. 
\end{align}
Noting the topological number stemming from the integral, 
\begin{align}
\lim_{\ep\to 0}\int_{\ga(\ep)}
d\theta =\theta(1)=2\pi n\,,  
\end{align}
we have the following proposition. 

\begin{prop}
The geometric phase is obtained by the Berry connections 
given in \eqref{eq:berry} as 
\begin{align}
\De_g
&=-i \lim_{\ep\to0}\int_{\ga(\ep)} 
\left(\vev{\psi_+|d|\psi_+}+\vev{\psi_-|d|\psi_-}\right)
\nln
&=-2i \lim_{\ep\to0}\int_{\ga(\ep)} 
\vev{\psi_+|d|\psi_+}-2\pi n  
\nln
&=-2i \lim_{\ep\to0}\int_{\ga(\ep)} 
\vev{\psi_-|d|\psi_-}+2\pi n  \,,   
\label{eq:Dg-berry}
\end{align}
where $n$ is the topological number given in \eqref{eq:top}\,,
and $\ga(\ep)\in S^2-\{(0,0,\pm1)\}$ is the regularized curve
in \eqref{eq:ga-reg}\,. 
\end{prop}

Hence, the geometric phase $\De_g$ of our model 
could be identified with the Berry phase for the two-level
system. More precisely, the Gauss vector $\bg$ in 
$S^2-\{(0,0,\pm1)\}$ corresponds to the eigenstates
$\ket{\psi_\pm}$ in the {\it Bloch sphere}, 
and one-form $\cos\be d\theta$ can be interpreted 
as the Berry connection $\vev{\psi_\pm |d|\psi_\pm}$\,. 
See \tabref{tab:berry}.

\begin{table}[htbp]
\centering
\setlength{\tabcolsep}{5pt}
\caption{The correspondence between our model 
and the Berry phase}
\begin{tabular}{|c|c|} \hline
Our model & Berry phase \\ \hline\hline 
$S^2$ & Bloch sphere \\ \hline
$\bg \in 
S^2-{(0,0,\pm1)}$ 
& $\ket{\psi_\pm}\in \CC^2$ \\ \hline
$\frac{i}{2}(\pm1+\cos\be) d\theta$ & 
$\vev{\psi_\pm |d|\psi_\pm}$
\\ \hline
\end{tabular}
\label{tab:berry}
\end{table}%

\section{Concluding remarks}
\label{sec:concl}
\setcounter{equation}{0}

In this article, we have focused on a simple kinematic
toy model exhibiting the geometric phases. 
That is, a rotating disc around a fixed disc without
slipping on the edge. 
The question is how much the moving disc has rotated. 
The total rotation angle consists of the dynamical phase $\De_d$ 
and the geometric phase $\De_g$\,. 
The latter is given by \thmref{thm:main}, 
and the answer to the question is given in \eqref{eq:ans}\,.  
To verify these claims, 
the Baumkuchen lemma\,\lemref{lem:bk} has played a crucial role. 
Furthermore, with the use of the Gauss vector \eqref{eq:gauss_a} 
and the regularized trajectory \eqref{eq:ga-reg}\,, 
we have obtained the main theorem \thmref{thm:main}. 
In \secref{sec:ex}\,, we have demonstrated how our theory 
fits some concrete cases. 

\ms 

In \secref{sec:com}\,, we have discussed the important models 
sharing the common structure, which include
Foucault's pendulum, Dirac's monopole potentials,
and Berry phases.   
This fact implies that the underlying mathematical structure
is the {\it Hopf fibration,}
\begin{align}
S^1 \hookrightarrow S^3 \xrightarrow{\pi} S^2\,, 
\end{align}
which is the principal fiber bundle equipped with  
$U(1)$ fiber embedded in the total space $S^3$
and the projection called the {\it Hopf map}
$\pi: S^3 \to S^2$\,.  
In this picture, the geometric phase $\De_g$ will be given by 
the holonomy of Hopf fibration. 
We shall report this aspect in the near future.

\section*{Acknowledgment}

We appreciate 
the Osaka Central Advanced Mathematical Institute, Osaka Metropolitan University, and
the Graduate School of Mathematics, Nagoya University 
as the hosts of 
the 32nd Japan Mathematics Contest and 
the 25th Japan Junior Mathematics Contest 2022, in which 
we have proposed our model as one of the problems \cite{JMC}.
We are grateful to our colleagues, 
Professors
Yoshiyuki Koga, 
Mitsutaka Kumakura, 
Yuki Sato,  
and 
Hiroshi Wakui
for their communications with the subject of this article. 
TM also thanks 
Professor Yuji Satoh
for his careful reading of 
the manuscript and valuable comments. 
The work of TM was supported by JSPS KAKENHI Grant Number JP19K03421 and JP24K06665. 

\appendix 
\renewcommand{\theequation}{\Alph{section}.\arabic{equation}}

\section{Alternative proof of Theorem \ref{thm:main}}
\label{app:alt-proof}
\setcounter{equation}{0}

Theorem \ref{thm:main} can also be proved by directly 
evaluating the right hand side of 
the line integral \eqref{eq:gp-reg}\,,  
\begin{align}
\De_g
=\lim_{\ep\to0}\int_{\ga(\ep)} \cos\be_\ep d\theta\,. 
\end{align}
Here, we only consider the case that $\ga$ is simple and 
surrounding the north pole $(0,0,1)$ on the left side, that is,
\begin{align}
(I_+, I_-)=(1,1)\,. 
\end{align}
The proof for the other cases is almost parallel.  
In this case, for sufficiently small $\ep>0$\,, 
we may assume that $\ga(\ep)=\ga$\,. 
Thus, it is enough to evaluate 
\begin{align}
\De_g
=\int_{\ga} \cos\be d\theta\,. 
\end{align}

\ms 

Because the local coordinate $(\theta, \be)$
is degenerated at $\be=0, \pi$\,, 
we need to remove these poles to apply 
the Stokes' theorem.  
For this purpose, we consider the sufficiently small disc 
$\de$ covering the north pole and decompose the integral
contour as follows (see Fig.\,\ref{fig:stokes}), 
\begin{align}
\int_{\ga} \cos\be d\theta
=\int_{\ga-\partial \de} \cos\be\, d\theta
+\int_{\partial \de} \cos\be\, d\theta\,. 
\label{eq:delta}
\end{align}
\begin{figure}[htbp]
\centering
\includegraphics[width=8cm]{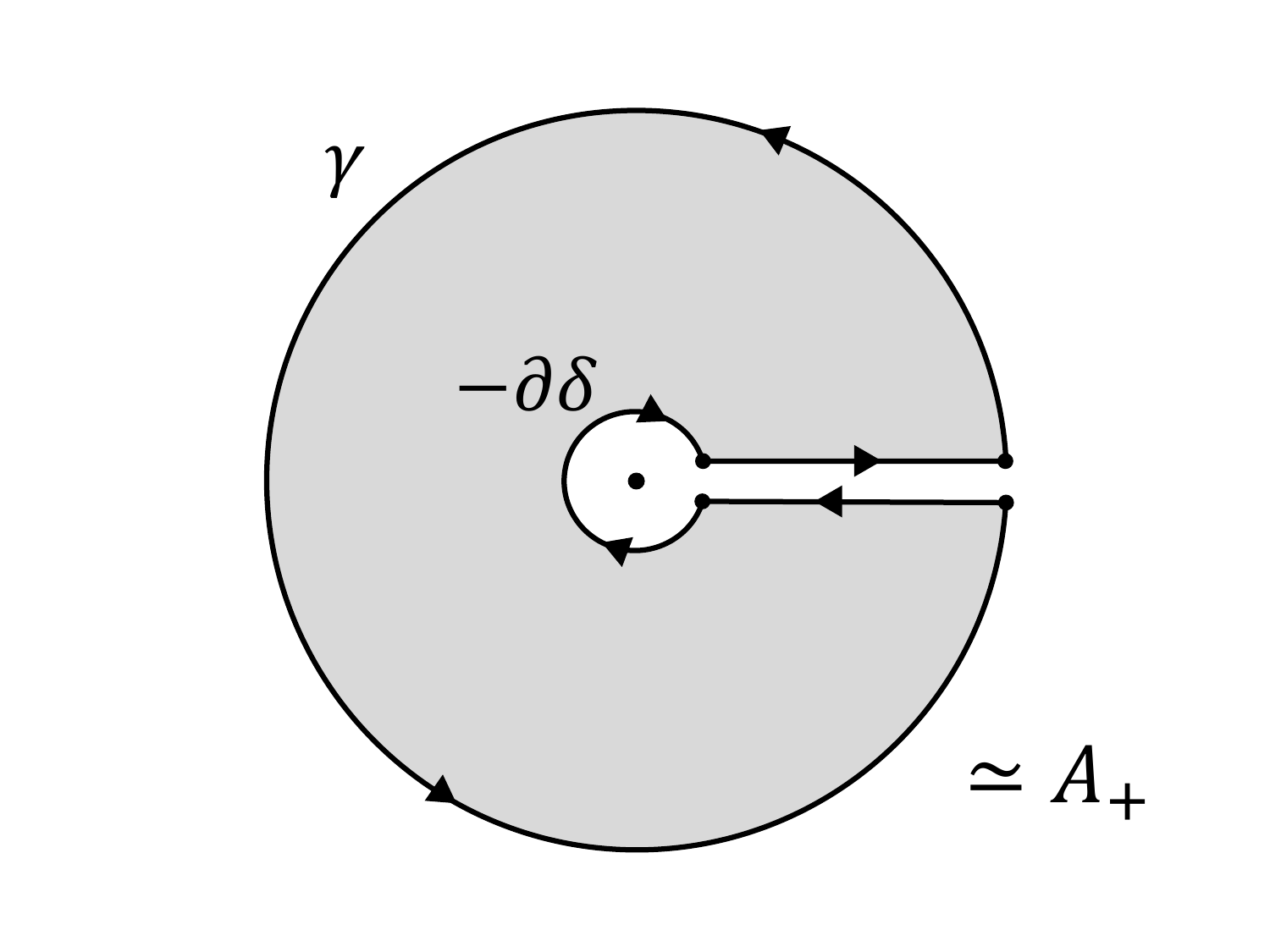}
\caption{The contour $\ga-\partial \de$\,. The area in gray 
approaches to $A_+$ as $\de$ shrinks to the north pole. }
\label{fig:stokes}
\end{figure}%

\ms 

The first terms on the right hand side 
can be evaluated by the Stokes' theorem as 
\begin{align}
\int_{\ga-\partial \de} \cos\be\, d\theta
=\int_{S_+-\de}d(\cos\be\, d\theta) 
=\int_{S_+-\de} \sin\be \,  d\theta\wedge d\be\,. 
\end{align}
This gives the area of $S_+-\de$ 
and reduces $A_+$ 
as $\de$ retracts to $(0,0,1)$\,. Hence, it holds that 
\begin{align}
\int_{\ga-\partial \de} \cos\be\, d\theta
\simeq A_+ \qquad \text{as} \qquad \de\to (0,0,1)\,. 
\end{align}
%
%
On the other hand, the second term is calculated as 
\begin{align}
\int_{\partial \de} \cos\be\, d\theta
\simeq - \int_{\partial \de}  d\theta=-2\pi\,,  
\end{align}
because $\be\simeq \pi$ when $\de$ is sufficiently small 
region around $(0,0,1)$\,.

\ms 

Since the decomposition \eqref{eq:delta} is independent of 
the size $\de$\,, 
it turns out to be 
\begin{align}
\int_{\ga} \cos\be d\theta
=A_+-2\pi\,. 
\end{align}
This is the desired result for the present case, 
\begin{align}
\De_g=A_+-2\pi I_+
\qquad \text{with}\qquad 
I_+=1\,. 
\end{align}

\section{Baumkuchen angle and the geodesic curvature}
\label{app:baum}
\setcounter{equation}{0}


It is worth clarifying the relation between Baumkuchen lemma \ref{lem:bk} and the geodesic curvature in \defref{def:curv-def}\,. 
For this purpose, we consider a ribbon-like neighborhood of the curve $\ga$. 
First, let us introduce a curve $\gamma_q$, which is locally parallel to $\gamma$ but separated by a small distance $q$ in the normal direction. 
Namely, we define
\begin{equation}
\gamma_q=\{\, g(s)-q\nu(s)\in\mathbb{R}^3 
\mid s\in[0,L] \,\}\subset\mathbb{R}^3,
\label{eq:ga-q}
\end{equation}
where $L$ is the length of $\gamma$. 
Here $q$ represents a small displacement parameter in the normal direction. 
Note that $\gamma_q$ is regarded as a curve in a ribbon-like neighborhood of $\gamma$, not necessarily lying on $S^2$. 
In the case of $S^2$, this expression corresponds to the first-order approximation 
of the exact spherical formula $(\cos q)g(s)-(\sin q)\nu(s)$.
We thank the referee for pointing this out.
Then, the length of the curve $\ga_q$\,, 
which we denote by $L_q$\,,  is given by 
\begin{align}
L_q
= \int_0^L |\bg'(s)-q\bnu'(s)| ds \,. 
\label{eq:lq}
\end{align}
Note that $L_0=L$ since $\ga_0=\ga$\,.

\begin{prop}
\label{prop:baum-geo}
Suppose that the curve $\ga$ is smooth. Then, it holds that 
\begin{align}
\lim_{q\to 0} \frac{L_q-L_0}{q} 
= \int_0^L \ka_g(s) ds \,. 
\label{eq:baum-geo}
\end{align}
\end{prop}

\ms 

\begin{proof}
Substituting \eqref{eq:lq} into the left hand side, it is calculated as 
\begin{align}
\lim_{q\to 0} \frac{L_q-L_0}{q}\Big|_{q=0}
&= \frac{d}{dq} L_q \big|_{q=0} 
\notag \\
&= \frac{d}{dq} \int_0^L |\bg'(s)-q\bnu'(s)| ds \Big|_{q=0} 
\notag \\
&= \int_0^L \frac{\partial }{\partial q} |\bg'(s)-q\bnu'(s)| \Big|_{q=0} ds  
\notag \\
&= \int_0^L \frac{-\bg'(s)\cdot \bnu'(s)+q \bnu'(s)\cdot \bnu'(s)}
{|\bg'(s)+q\bnu'(s)|}  \Big|_{q=0} ds  
\notag \\
&= - \int_0^L \bg'(s)\cdot \bnu'(s)  ds  \,. 
\end{align}
where we have used \dis{|\bg'(s)|=1} in the last equality. 
Due to the relation \eqref{eq:gvgv} and the definition of $\ka_g$ 
in \eqref{eq:curv-def}\,, we obtain the desired relation  
\begin{align}
\lim_{q\to 0} \frac{L_q-L_0}{q}\Big|_{q=0}
=\int_0^L \bg''(s)\cdot \bnu(s)  ds  
=\int_0^L \ka_g(s)  ds   \,. 
\end{align}
This proves the proposition. 
\end{proof}

\ms 

\begin{remark}
The relation \eqref{eq:baum-geo} could be regarded as 
the generalization of Baumkuhen lemma \ref{lem:bk}\,.  
To compare with \lemref{lem:bk}\,, suppose that 
the curve $\ga$ is a part of a circle, namely an arc, 
and the distance between two arcs $q$ 
is smaller than the radius of the circle, $r$\,. 
Then, \eqref{eq:baum-geo} becomes 
\begin{align}
\frac{L_q-L_0}{q}=\ka_g L  
\label{eq:baum-geo-ex}
\end{align}
as $\ka_g$ is constant. 
Notice that, by Baumkuchen lemma \ref{lem:bk}\,, the left hand side is
equal to the angle $\theta_B$ of 
a piece of  Baumkuchen, {\it i.e.}\,,
\begin{align}
\frac{L_q-L_0}{q}=\frac{(r-q)\theta_B-r \theta_B}{q}
=-\theta_B\,. 
\end{align}
As $L=L_0 =r \theta_B$ on the right hand side\,, 
the equation \eqref{eq:baum-geo-ex} turns out to be  
\begin{align}
\ka_g=-\frac{1}{r}\,. 
\end{align}
This is indeed the curvature of a circle with radius $r$ 
up to the sign. This discrepancy is just owing to our conventions for  
the Baumkuchen angle $\theta_B$ and the geodesic curvature 
$\ka_g$\,.  
Hence, \propref{prop:baum-geo} is the natural generalization
of Baumkuchen lemma \ref{lem:bk}\,. 
\end{remark}

\ms 

\begin{remark}
The relation \eqref{eq:baum-geo} holds not only  for $\ga\subset S^2$ 
but also for an arbitrary smooth curve 
in a two-dimensional smooth surface in $\RR^3$\,. 
\end{remark}


\ms 

At the end of this section, let us apply \propref{prop:baum-geo}
to \thmref{thm:De-twoint}\,. 

\begin{cor}
\label{cor:fou}
Suppose that $\ga$ is a simple closed curve and smooth, and 
winding the north pole $(0,0,1)$ on the left side, {\it i.e.}\,, 
$(I_+, I_-)=(1,1)$\,. Then, the geometric phase is given by  
\begin{align}
\De_g=\lim_{q\to 0} \frac{L_0-L_q}{q} \,. 
\end{align}
Furthermore, when $\be$ is constant, it reduces to 
\begin{align}
\De_g
=2\pi \cos\be\,. 
\label{eq:Dg-cos}
\end{align}
\end{cor}

\begin{proof}
Plugging \eqref{eq:baum-geo} with \eqref{eq:De-twoint}\,, 
we have 
\begin{align}
\De_g
=\lim_{\ep\to0} \int_{\ga(\ep)} d\vp 
- \lim_{q\to 0} \frac{L_q-L_0}{q} 
- \sum_{i=1}^n \al_i \,. 
\end{align}
When $\ga$ is a simple closed curve and smooth, and $(I_+, I_-)=(1,1)$\,,
the first and third terms vanish. 
Then, it becomes  
\begin{align}
\De_g=\lim_{q\to 0} \frac{L_0-L_q}{q} \,. 
\end{align}
In addition, when $\be$ is constant, $\ga$ is a circle 
of radius $\sin\be-q\cos\be$\,. 
In this case, we get 
\begin{align}
L_q=2\pi (\sin\be-q\cos\be)\,. 
\end{align}
Substituting this into the above, we get 
\begin{align}
\De_g
=\lim_{q\to 0} \frac{L_0-L_q}{q} 
=2\pi \cos\be\,. 
\end{align}
\end{proof}

\ms 

It is noted that the expression \eqref{eq:Dg-cos}
coincides with \eqref{eq:exam-iv-geo} of the example (iv) 
that we have discussed in \secref{sec:ex}\,. 
Since the example (iv) corresponds to Foucault's pendulum 
discussed in \secref{subsec:fou}\,, 
\corref{cor:fou} also provides a simple proof 
of the Foucault's sine law \eqref{eq:fou-sin}
with the appropriate identifications in \tabref{tab:fou}\,.


\end{document}